\documentclass[a4paper,abstracton]{scrartcl}
\usepackage{amsfonts}
\usepackage{amssymb}
\usepackage{amsmath}
\usepackage{amsthm}
\usepackage{algorithm}
\usepackage{algorithmic}
\usepackage{graphicx}
\usepackage{color}

\usepackage{subcaption}

\numberwithin{equation}{section} 
\numberwithin{table}{section}

 \newtheorem{lemma}{Lemma}[section]
 \newtheorem{prop}{Proposition}[section]
  \newtheorem{remark}{Remark}[section]
  
\newtheorem{theorem}{Theorem}[section]
\newtheorem{example}{Example}[section]

\usepackage{authblk}

\begin{document}

\newcommand{\rz}{{\mathbb{R}}}
\newcommand{\nz}{{\mathbb{N}}}
\newcommand{\zz}{{\mathbb{Z}}}
\newcommand{\eps}{\varepsilon}
\newcommand{\cei}[1]{\lceil #1\rceil}
\newcommand{\flo}[1]{\left\lfloor #1\right\rfloor}
\newcommand{\opt}{\mbox{\sc OPT}}
\newcommand{\sbb}{\mbox{sub-block }}

\title{The Multi-level Bottleneck Assignment Problem: Complexity and Solution Methods}

\author[1]{Trivikram Dokka}
\author[2]{Marc Goerigk}

\affil[1]{ Department of Management Science, Lancaster University, United Kingdom
  (\texttt{t.dokka@lancaster.ac.uk})}
\affil[2]{Network and Data Science Management, University of Siegen, Germany
 (\texttt{marc.goerigk@uni-siegen.de})}

\date{}

\maketitle

\begin{abstract}
We study the multi-level bottleneck assignment problem (MBA), which has important applications in scheduling and  quantitative finance. Given a weight matrix, the task is to rearrange entries in each column such that the maximum sum of values in each row is as small as possible. We analyze the complexity of this problem in a generalized setting, where there are restrictions in how values in columns can be permuted. We present a lower bound on its approximability by giving a non-trivial gap reduction from three-dimensional matching to MBA. 

To solve MBA, a greedy method has been used in the literature. We present new solution methods based on an extension of the greedy method, an integer programming formulation, and a column generation heuristic. In computational experiments we show that it is possible to outperform the standard greedy approach by around 10\% on random instances.

\end{abstract}


\section{Introduction}

\subsection{Problem definition and state-of-the-art}
\label{lit_review}

The following axial assignment problem arises to scheduling, rostering and finance applications: Given  are $m$ pairwise disjoint sets $S_1, S_2, \ldots, S_m$ each of cardinality $n$, and a weight $w(s) \in \nz$ for each $s \in S$  where $S = \cup_{i\in[m]} S_i$ and $[m]:=\{1,\ldots,m\}$. The set $S$ can be seen as the node-set of an $m$-partite graph that has a given set of arcs $E=\bigcup_{i\in[m-1]} E_i$, where $E_i=\{(u,s)|~u \in S_i, s \in S_{i+1} \}$ connects nodes from $S_i$ with nodes from $S_{i+1}$. An $m$-tuple $D=(s_1, s_2, \ldots, s_m)$ is feasible if $s_i \in S_i$ for $i \in[m]$ and $(s_i,s_{i+1}) \in E$. The weight of an $m$-tuple $D$ equals $w(D) = \sum_{s \in D} w(s)$. The problem is to find a partition of $S$ into $n$ feasible $m$-tuples $D_1, D_2, \ldots, D_n$ such that $\mbox{max}_{j\in[n]} w(D_j)$ is as small as possible. We refer to this partition of $S$ into $\{D_1, D_2, \ldots, D_n\}$ as a solution $M$, and the weight $w(M)$ of a solution $M=\{D_1, D_2, \ldots, D_n\}$ equals $\mbox{max}_{j\in[n]} w(D_j)$. This problem is known as the multi-level bottleneck assignment problem (MBA). It is often seen through the lens of column permutation in a matrix (under constraints), such that the maximum row sum is minimized. For the sake of clarity, since we discuss column generation based algorithms in Section~\ref{sec:methods}, we stick to the convention of sets and tuples in referring to MBA.

In the following, we discuss the motivation and current literature on this problem. There are two main application areas: scheduling, and finance.

MBA was first introduced and studied by Carraresi and Gallo~\cite{cargal}, motivated by an application in bus driver scheduling. Special cases of the problem have been studied even before \cite{cargal}. 
A particularly important special case which we call {\em complete-MBA}, as referred to in Dokka et al.~\cite{DKS2012}, is when each $E_i$ is complete. 
The approximability of this special case has been studied by Hsu~\cite{hsu} and by Coffman and Yannakakis~\cite{cofyan}. For complete-MBA, Hsu~\cite{hsu} gave an ($2-\frac{1}{n}$)-approximation algorithm that runs in $O(m n \mbox{log}n)$, while Coffman and Yannakakis~\cite{cofyan}  gave an ($\frac32 - \frac{1}{2n}$)-approximation algorithm that runs in $O(n^2m)$. For the case where $m=3$, Hsu gave a $\frac32$-approximation algorithm that runs in $O(n \mbox{log} n)$, and a $\frac43$-approximation algorithm that runs in $O(n^3\mbox{log}n)$. 

Another important problem that MBA contains as a special case is the bi-criteria scheduling problem in which one tries to find a schedule with minimum makespan over all flow time optimal schedules on identical machines. This problem was first studied in 1976 by Coffman and Sethi~\cite{cofsethi}, where a $\frac54$ approximation algorithm is given. Eck and Pinedo~\cite{pinedo} give a $\frac{28}{27}$ approximation for the two machine case. More recently Ravi et al.~\cite{ravi} prove the Coffman and Sethi conjecture on the performance of (a natural extension) of the longest processing time algorithm to this bi-criteria scheduling problem. We note that the greedy algorithm often used for MBA when applied to the above special case of bi-criteria scheduling can be interpreted as an extension of the longest processing time list scheduling algorithm.

MBA is also connected to parallel machine scheduling with bags, see Das and Wiese~\cite{Das_ESA17} and  Page and Oba~\cite{page_solis-Oba_AAIM2018}. In this problem jobs which belong to the same bag cannot be scheduled on the same machine, while in our case we have the restriction that only jobs connected by an edge in the underlying layered graph can be scheduled together on same machine. 

\smallskip

The problem also has important applications in quantitative finance where one needs to infer the stochastic dependence between many random variables. More specifically, we are given a set of random variables $L_i$, $i\in[m]$, with known marginal distributions. However, little or nothing is known about the dependence structure between these variables and about the distribution of an aggregate random variable $L$ obtained, for example, by summing all $L_i$. Understanding the distribution of $L$ or inferring the dependence between $L_i$ has important implications in practice, such as a better understanding of the overall system risk for risk managers, or a better idea of overall portfolio risk for portfolio managers. In Embrechts et al.~\cite{embrechts_JBF13} an application of MBA is given in estimating the upper and lower bounds of the Value-at-Risk over all possible dependence structures. The problem is particularly related to the concepts of complete mixability and joint mixability, see Wang and Wang~\cite{wang2_JMA11, wang2_MOR16}; and to the minimum variance problem, see Ruschendorf~\cite{Ruschendorf1983_metrika}. Bernard et al.~\cite{bernard2018_aor} study the problem in inferring the dependence among variables $L_i$. 
In Puccetti and Ruschendorf~\cite{puccetti_JCAM12} and Hsu~\cite{hsu} a natural greedy heuristic, which is called \emph{Rearrangement Algorithm (RA)}, was studied and analysed. Variants of RA have received much attention in the recent years, see Boudt et al.~\cite{boudt2015} and Bernard et al.~\cite{bernard2015a}. Within this domain of research, more recently, Bernard et al.~\cite{bernard2018_wp1807} investigate joint distributions when not just marginals of individual random variables but also marginal distributions of a linear combinations of subsets of variables is known. Our work in this paper deals with more general case compared to the case studied in Bernard et al.~\cite{bernard2018_wp1807}.

As far as we are aware, all known approximation results except Dokka et al.~\cite{DKS2012} deal with the complete (all $E_i$ are complete bipartite). In this paper we deal with a more general setting, namely the case where the edge set between $S_i$ and $S_{i+1}$ can be arbitrary (and not necessarily complete)  for $i\in[m-1]$. 

Related problems have also been studied from an approximation point of view. For example, one can see MBA as a generalization of the classical multi-processor scheduling problem with incompatibilities between jobs. Such related problems have been studied in Bodlaender et al.~\cite {bodjanwoe}. Other types of (three-dimensional) bottleneck assignment problems have been studied by Klinz and Woeginger~\cite{kliwoe} and Goossens et al.~\cite{goopol}. 

\subsection{Research questions and contributions}

While considerable work is dedicated to understanding the complexity of MBA when all edge sets are complete, not much is known about the case when the edge sets are arbitrary. The only known result is from Dokka et al.~\cite{DKS2012}, which gives a lower bound of $2\cdot OPT$ on approximability in the case when $m=3$ and shows a simple greedy approach to achieve a matching upper bound. Complete-MBA is shown to admit a PTAS, first shown in \cite{DKS2012}, while the simple greedy approach is already known to yield a 2-approximation in Hsu~\cite{hsu}. Given that the greedy approach gives a constant factor approximation in the complete case, it is tempting to believe a similar approach may be used to construct a constant factor approximation for the general case, which leads to our first question. 

\smallskip

\textit{Question 1:} Is there a polynomial time algorithm which approximates MBA to within a constant factor of the optimum objective value?

\smallskip
The simple greedy algorithm, however bad in the worst case, may still be desirable owing to the existence of efficient algorithms for the bottleneck assignment problem. On the other hand integer programming solvers such as Cplex have undergone massive improvements in the last years, which led to efficient practical heuristics based on mathematical programming formulations for some hard problems. So we ask the following second question.

\smallskip

\textit{Question 2:} Can the performance of the basic greedy algorithm be improved, and how does it compare with more advanced matheuristics based on integer programming (IP) formulations?


We answer the above research questions with the following results:

\begin{itemize}
\item  We show that the answer to first question is NO under the assumption of $P\neq NP$ by giving a non-trivial gap reduction from three-dimensional matching (3DM). More specifically, we prove that the existence of a $\{(u+1)-\epsilon\}$-polynomial time approximation algorithm for MBA with $m=3u$ implies $P=NP$.
\item We show that extensions of the greedy method in combination with mathematical programming techniques can lead to significantly better solutions within the same time limit than when using the standard greedy method.
\end{itemize}

The rest of the paper is structured as follows. We give the inapproximability result for the general case in Section~\ref{inapp}. In Section~\ref{sec:methods}, we present an integer programming formulation along with a greedy and a column generation heuristic. These methods are compared experimentally using random MBA instances in Section~\ref{sec:experiments}. Section~\ref{sec:conclusions} summarizes our findings and points out further research questions.

\section{Inapproximability of the arbitrary case}\label{inapp}

In Hsu~\cite{hsu} it is shown that for complete-MBA the natural sequential heuristic achieves a 2-approximation. It is tempting to believe that this may be true even in the arbitrary case. We show that MBA for a fixed $m>3$ cannot be approximated within a factor of $\lfloor\frac{m}{3}\rfloor+1$ unless $P=NP$.
To do so, we show that a YES-instance of
3-dimensional matching (3DM) corresponds to an instance of MBA with weight 1, whereas a
NO-instance corresponds to an instance of our problem with weight $\frac{m}{3}+1$. Then, a
polynomial time approximation algorithm with a worst case ratio strictly less than $\frac{m}{3}+1$
would be able to distinguish the YES-instances of 3DM from the NO-instances,
and this would imply $P = NP$.

Let us first recall the 3-dimensional matching problem:

Instance: Three sets $X = \{x_1, . . . , x_q\}$, $Y = \{y_1, . . . , y_q\}$, and $Z =
\{z_1, . . . , z_q\}$, and a subset of tuples $T \subseteq X \times Y \times Z$.

Question: does there exist a subset $T'$ of $T$ such that each element of $X \cup Y \cup Z$ is in exactly one triple of $T'$?\\

Let the number of triples be denoted by $|T| = p$. Further, let the number of triples in which element $y_{\ell}$ occurs be denoted by $\#occ(y_{\ell})$, $\ell=1, \ldots, q$.

Starting from arbitrary instance of 3DM, we now build a corresponding instance of MBA by specifying sets $S_i$, edges $E$, and the weights $w$. Before we explain the construction we first explain the basic building blocks and gadgets which are pieced together to form a MBA instance.

\begin{example}
To illustrate our construction, we use the following example instance of 3DM:
\begin{align*}
& X = \{ x_1,x_2\},\ Y=\{y_1,y_2\},\ Z=\{z_1,z_2\} \\
& T = \{ t_1, t_2, t_3\} \\
& t_1=(x_1,y_1,z_1),\ t_2=(x_2,y_2,z_2),\ t_3=(x_1,y_2,z_1) \\
& q = 2,\ p = 3,\ d = p-q = 1 \\
& \#occ(y_1) = 1,\ \#occ(y_2) = 2
\end{align*}
\end{example}

\subsection{Building sub-blocks}

There are two types of nodes in the resulting MBA instance, which we call \emph{main} and \emph{dummy} nodes. 

The main nodes in each set in the MBA instance are partitioned into sub-blocks of nodes. Each \sbb is of cardinality $q$, $p$, or $d :=p-q$. We use the following types:
\begin{itemize}
\item $X$-sub-blocks, where each node corresponds to one element in $X$ (cardinality $q$)
\item $Z$-sub-blocks, where each node corresponds to one element in $Z$ (cardinality $q$)
\item $Y$-sub-blocks, where $\#occ(y_i)-1$ many nodes correspond to one of each element $y_i\in Y$ (cardinality $d$)
\item $T$-sub-blocks, where each node corresponds to one triple in $T$ (cardinality $p$)
\end{itemize}

In our construction, we may refer to two sub-blocks as being \emph{connected}. When this is the case, the corresponding edge set depends on the type of sub-blocks:
\begin{itemize}
\item $X$-sub-blocks are connected to $T$-sub-blocks by connecting a node in the $X$-sub-block corresponding to an element $x_i\in X$ with those nodes in the $T$-sub-block corresponding to tuples that contain $x_i$
\item $Y$- and $Z$-sub-blocks are connected to $T$-sub-blocks in the same way
\item two $T$-sub-blocks are connected by connecting each two nodes corresponding to the same tuple
\end{itemize}

The role of the dummy nodes and their exact number will be apparent when we explain the connections between gadgets in the later sections. 

\begin{example}[continued]
In our 3DM example, $X$- and $Z$- sub-blocks have $2$ nodes each. A $Y$-sub-block has only one node, corresponding to $y_2$. $T$-sub-blocks have three nodes. Figure~\ref{fig1} shows the different possibilities how these sub-blocks can be connected.

\begin{figure}[htb]
\begin{center}
\begin{subfigure}[t]{0.3\textwidth}
\includegraphics[width=\textwidth]{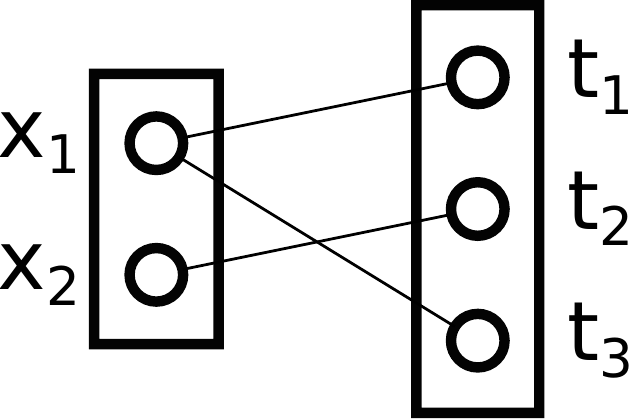}
\caption{Connecting $X$- and $T$-sub-blocks.}\label{fig1-1}
\end{subfigure}
\hspace*{1cm}
\begin{subfigure}[t]{0.3\textwidth}
\includegraphics[width=\textwidth]{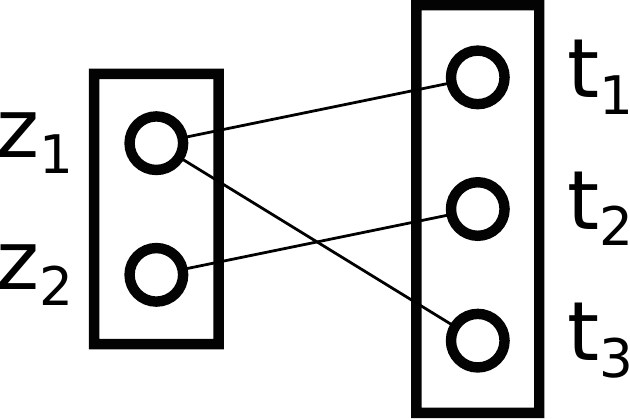}
\caption{Connecting $Z$- and $T$-sub-blocks.}\label{fig1-2}
\end{subfigure}\\
\begin{subfigure}[t]{0.3\textwidth}
\includegraphics[width=\textwidth]{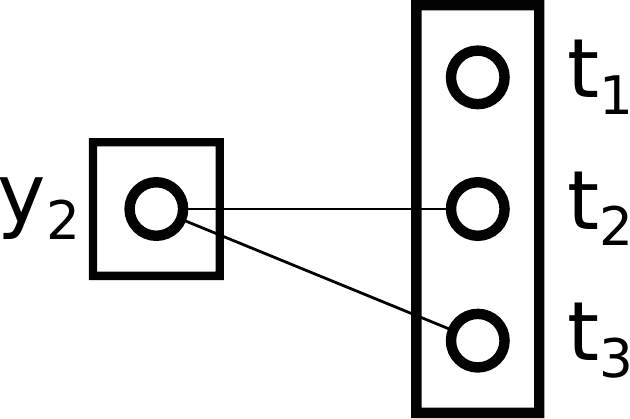}
\caption{Connecting $Y$- and $T$-sub-blocks.}\label{fig1-3}
\end{subfigure}
\hspace*{1cm}
\begin{subfigure}[t]{0.3\textwidth}
\includegraphics[width=\textwidth]{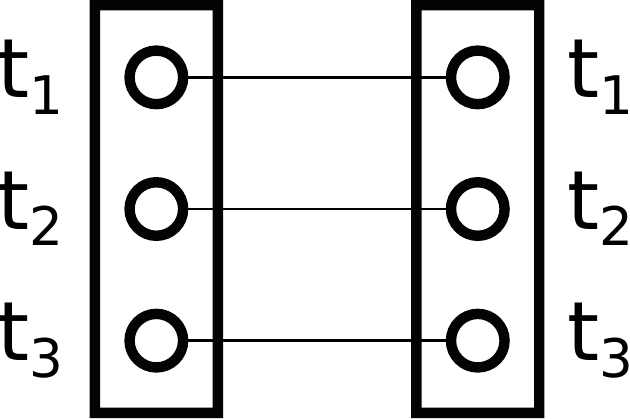}
\caption{Connecting $T$- and $T$-sub-blocks.}\label{fig1-4}
\end{subfigure}
\caption{Connecting sub-blocks in the example.}\label{fig1}
\end{center}
\end{figure}
\end{example}

\subsection{Gadget construction}
The construction is mainly based on two gadgets $G_0$ and $G_1$, which are both MBA instances with $m=3$. Each gadget consists of multiple blocks $j$, where each block is made up of three sets $V_1^{\#,j}$, $V_2^{\#,j}$, and $V_3^{\#,j}$ for $G_\#$, $\#\in\{0,1\}$. 

\paragraph{Blocks and connections within blocks:}
Each block $j$ of $G_{\#}$, $\# \in \{0,1\}$, has three times $2p$ nodes grouped as follows:
\begin{itemize}
\item in $V_1^{\#,j}$, there is a $Z$-sub-block, a $Y$-sub-block, and a $T$-sub-block
\item in both $V_2^{\#,j}$ and $V_3^{\#,j}$, there are two $T$-sub-blocks
\end{itemize}
The head $T$-sub-block in  $V_2^{\#,j}$ is connected to the $Z$-sub-block and $Y$-sub-block in  $V_1^{\#,j}$, and to the head $T$-sub-block in $V_3^{\#,j}$. The tail $T$-sub-block in  $V_2^{\#,j}$ is connected to the tail $T$-sub-block in $V_1^{\#,j}$, and to the head and tail $T$-sub-blocks in $V_3^{\#,j}$. 

\begin{example}[continued]
Figure~\ref{blocks_new} illustrates how blocks $V_1^{\#,j}$, $V_2^{\#,j}$, and $V_3^{\#,j}$ are constructed. 
\begin{figure}[htb]
\begin{center}
\begin{subfigure}[t]{0.4\textwidth}
\includegraphics[width=\textwidth]{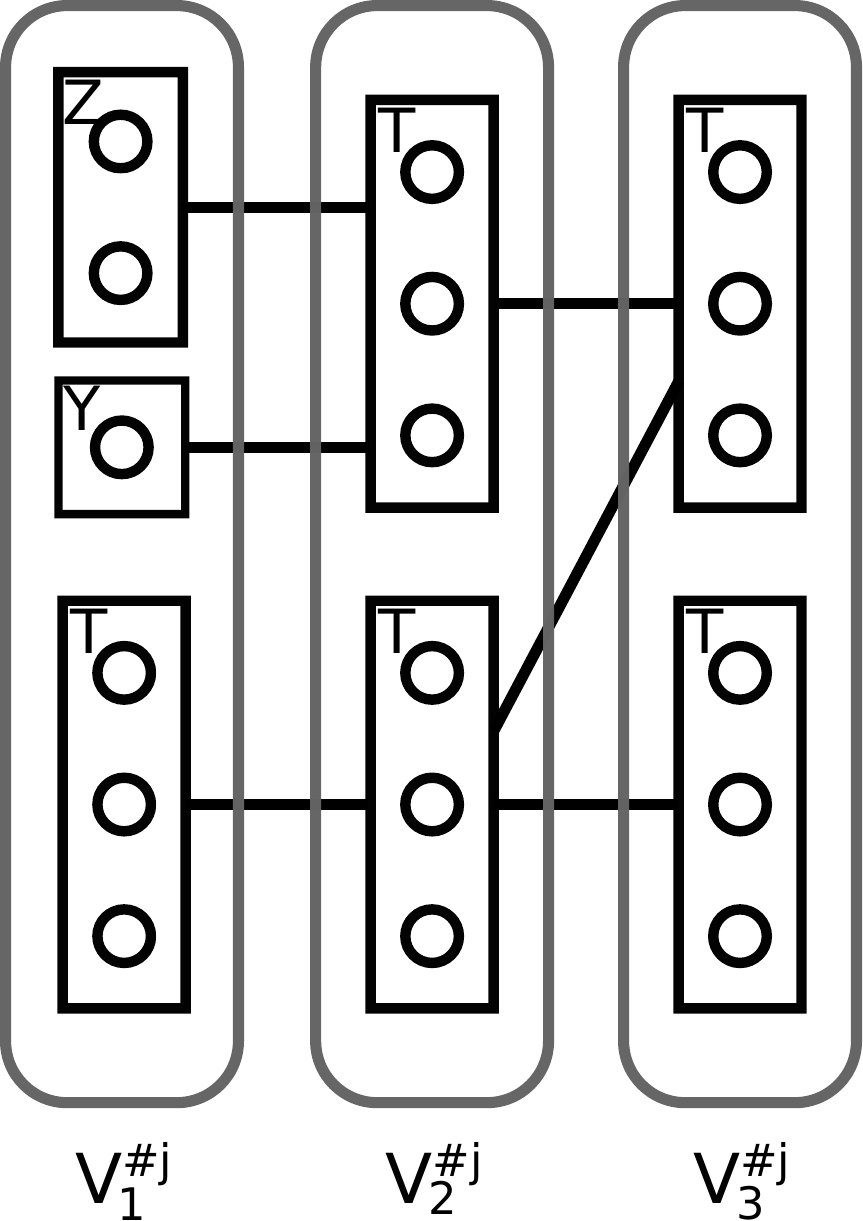}
\caption{Illustration of $j^{th}$ block in $G_0$ and $G_1$. A single connection represents multiple edges as indicated in Figure~\ref{fig1}.}\label{blocks_new}
\end{subfigure}
\hfill
\begin{subfigure}[t]{0.4\textwidth}
\includegraphics[width=\textwidth]{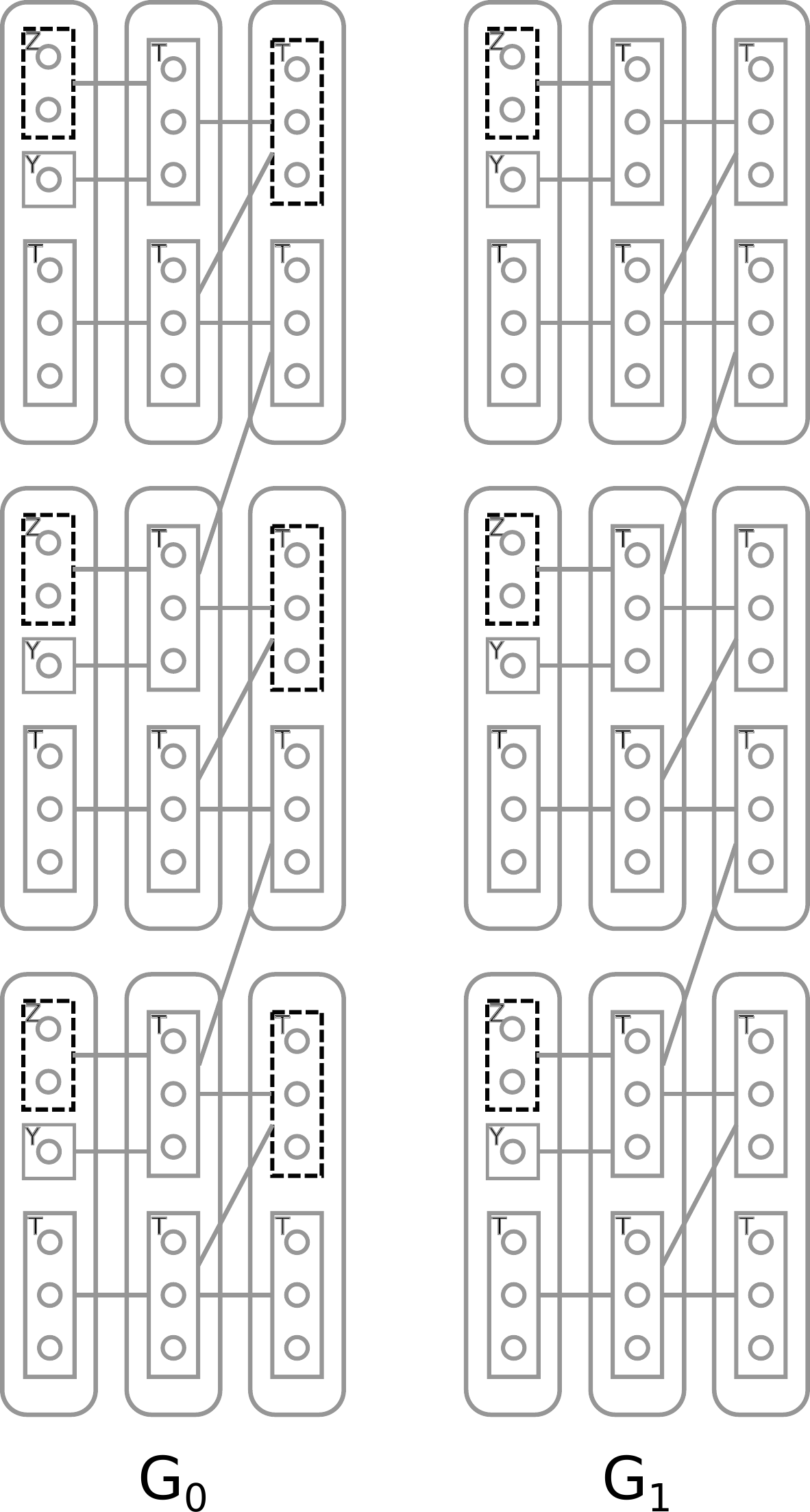}
\caption{Connections between blocks. The highlighted sub-blocks have nodes with node weights equal to $1$.}\label{blocks_new2}
\end{subfigure}
\caption{Illustration of block construction and connection.}
\end{center}
\end{figure}
\end{example}

\paragraph{Connections between blocks:} 
We denote by $height$ the number of blocks in a gadget. 
For $1\leq j < height$, the tail $T$-\sbb in $V_3^{\#,j}$ is connected to head $T$-\sbb of $V_2^{\#,j+1}$. 
Note that this way, only adjacent blocks are connected by edges.

\paragraph{Weights:}
Note that in the construction so far, gadgets $G_0$ and $G_1$ are identical. They differ with respect to their weights.
The head $T$-sub-block in $V_3^{0,j}$ and the $Z$-sub-block in $V_1^{0,j}$ and $V_1^{1,j}$ have all nodes weight equal to 1 for all $j$. All other nodes are weighted 0. 

\begin{example}[continued]
Figure~\ref{blocks_new2} shows how blocks are connected. Sub-blocks with weight 1 are highlighted with dashed lines.
\end{example}

\subsection{Main Construction}

\subsubsection{Overview}

We intend to show that it is not possible to approximate MBA within $\lfloor\frac{m}{3}\rfloor+1$ of optimum, for every fixed $m$. The exact structure of MBA instance constructed from 3DM depends on $m$. We assume that $m=3u$ for an integer $u$. We put together $(u-1)$ many $G_1$ gadgets and one $G_0$ gadget, each of different height, in that order to create an MBA instance. Since each gadget is itself an MBA instance with three columns, the resulting instance is a MBA instance with $m=3u$. For the ease of explanation we refer to each triple $(S_i, S_{i+1}, S_{i+2})$ as a layer, $i=1, 4, 7, \ldots, 3u-2$. The number of layers is equal to $u$. While we count the sets starting from left to right, we count layers from right to left. That is, layer $k+1$ is placed to the left of layer $k$. The height of the gadget in layer $k\in[u]$ is equal to $q^{k-1}p^{u-k}$ blocks. Before we explain how these gadgets are connected  in sequence we need some additional blocks of nodes as follows.

\begin{example}[continued]
In our example, we would like to use the 3DM instance to construct an MBA instance with $m=6$. This means that $u=2$. Hence, we have one $G_1$ gadget and one $G_0$ gadget. The height of the first ($G_0$) gadget is $q^0p^1 = 3$, and the height of the second ($G_1$) gadget is $q^1p^0 = 2$.
\end{example}

\subsubsection{Non-gadget nodes}

\paragraph{$X$-sub-blocks:}
We have one $X$-\sbb in every $S_i$ of each layer $k$. These are connected as follows:
\begin{itemize}
    \item the $X$-\sbb in $S_{3k}$ is connected to the head $T$-\sbb of $V_2^{\#,1}$ (first block of $G_{\#}$ in the $k^{th}$ layer) in $S_{3k-1}$,
     \item the $X$-\sbb in $S_{3k-1}$ is connected to the tail $T$-sub-block of $V_3^{\#,height}$ (in the last block of $G_{\#}$ in the $k^{th}$ layer) in $S_{3k}$,
      \item the $X$-\sbb in $S_{3k-2}$ is connected to the $X$-\sbb of $S_{3k-1}$ element-wise,
\end{itemize}
where \# takes a value 0 in layer $u$ and 1 for all other layers.

\paragraph{Dummy blocks:}
Each $S_i$ has an additional set of nodes apart from gadget and non-gadget sub-blocks. We refer to these nodes as \textit{dummy} nodes. The number of dummy nodes in layer $k$ is equal to sum of non-dummy nodes in all other layers. The dummy nodes in layer $k$ are evenly distributed over the three columns $S_{3k}$, $S_{3k-1}$ and $S_{3k-2}$. All dummy nodes in $S_{3k}$ are connected to all dummy nodes in $S_{3k-1}$, and all dummy nodes in $S_{3k-1}$ are connected to all dummy nodes in $S_{3k-2}$. All dummy nodes have a weight equal to 0.

\subsection{Connecting gadgets}

There are two types of edges connecting gadgets: Edges to dummy nodes, and edges between $Z$- and $T$-sub-blocks. Edges to dummy nodes in layer $k$ are constructed as follows:
\begin{itemize}
\item Every node from every $Y$- and $T$-sub-block in $S_{3k-2}$ is connected to every dummy node in $S_{3(k+1)}$ of the neighboring layer $k+1$.
\item Every node from every bottom $T$-sub-block in $S_{3k}$ is connected to every dummy node in $S_{3(k-1)-2}$ of the neighboring layer $k-1$.
\end{itemize}

Nodes in $Z$-sub-blocks of layer $k$ are connected to $T$-sub-blocks in $S_{3(k+1)}$ of
next layer $k+1$ in the following way. Blocks in layer $k$ are grouped into $p^{u-k}$ groups of $q^{k-1}$ blocks each. In each group, there are thus $q^{k-1}$ $Z$-sub-blocks with $q$ nodes each. 

Figure~\ref{zex} provides an example with $q=2$, $p=3$, and $u=3$. Layer $k=1$ has 9 blocks that are paritioned into 9 groups of one block each. Layer $k=2$ has 6 blocks that are partitioned into 3 groups of 2 blocks each. Finally, layer $k=3$ has 4 blocks that form a single group.

\begin{figure}[htb]
\begin{center}
\includegraphics[width=0.4\textwidth]{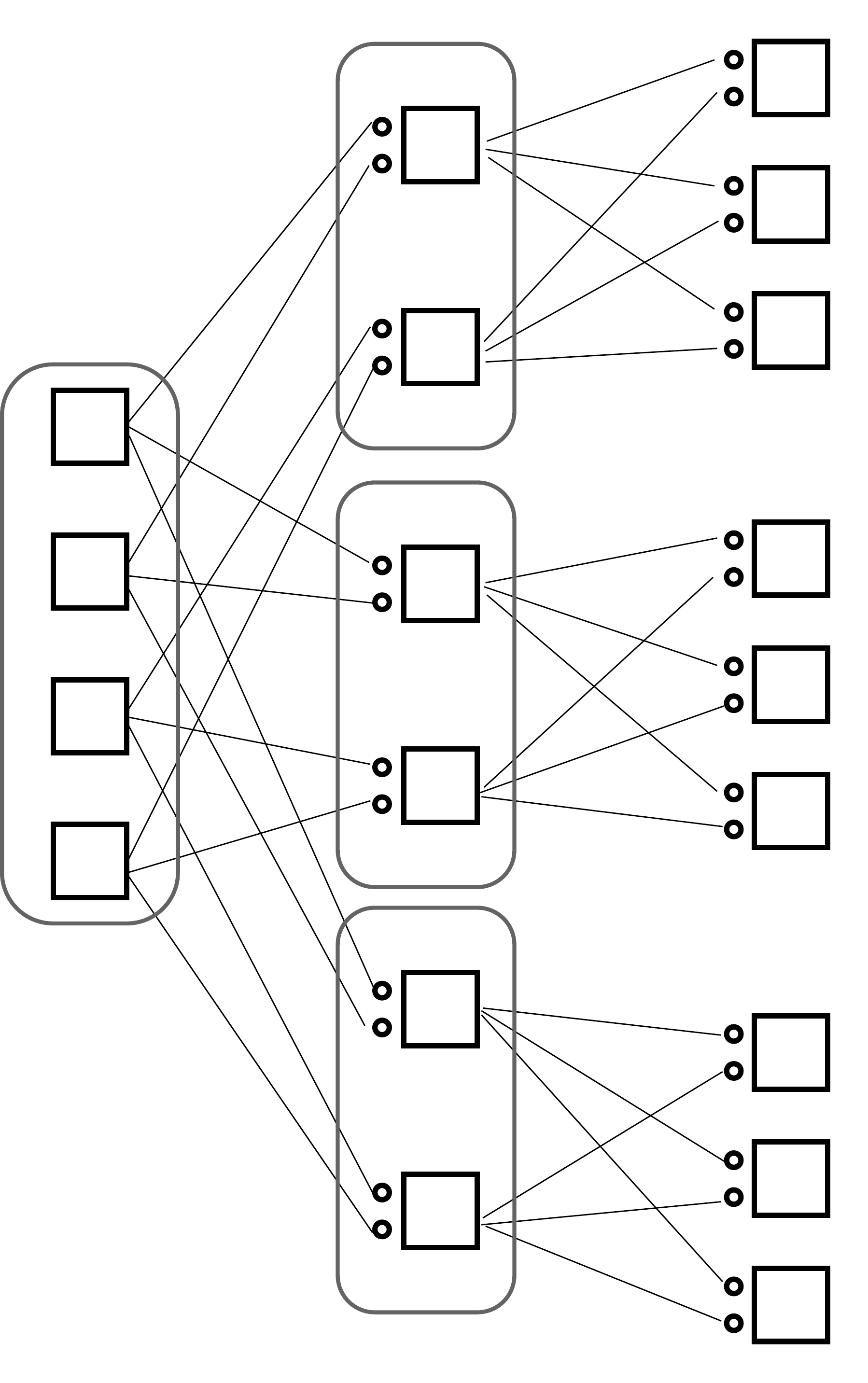}
\end{center}
\caption{Connecting gadgets by edges between $Z$- and $T$-sub-blocks.}\label{zex}
\end{figure}

Using $p$ groups at a time, we connect the $\ell$th $z$-node to all nodes of the $\ell$th top $T$-sub-block in the neighboring layer. In Figure~\ref{zex} we see that for the top $p$ blocks of the first layer, always the first node of each $Z$-block is connected to the top $T$-sub-block of the first block. The second node of each $Z$-block is connected to the top $T$-sub-block of the second block. In the first group of the middle layer, there are two $Z$-sub-blocks with two nodes each. These four nodes are connected to the four blocks of the third layer. In the same way, the four nodes of the second and the four nodes of the third group are connected to the four blocks of the third layer.

\begin{example}[continued]
We show the complete MBA instance resulting from the example 3DM instance in Figure~\ref{excomplete}. As there are two layers, $Y$- and $T$-sub-blocks on the left of the right layer are connected to dummy nodes on the right of the left layer. Bottom $T$-sub-blocks on the right of the left layer are connected to dummy nodes on the left of the right layer. $Z$-sub-blocks of the right layer are connected to $T$-sub-blocks on the right of the left layer as indicated.
\begin{figure}[htbp]
\begin{center}
\includegraphics[width=0.5\textwidth]{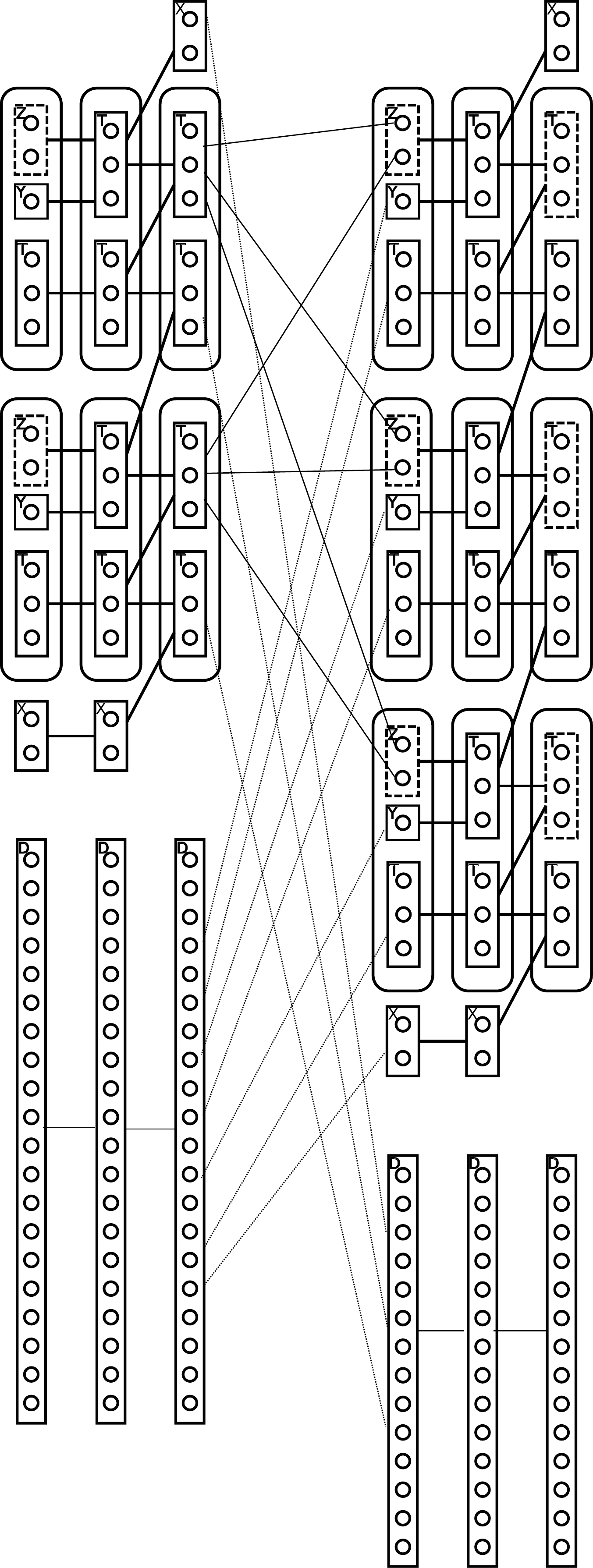}
\end{center}
\caption{Complete MBA instance.}\label{excomplete}
\end{figure}
\end{example}

\subsection{Instance analysis}

 We first note that the size of instance grows with $m$, but is bounded by a polynomial when $m$ is fixed.

\begin{lemma}\label{lemma1}
The number of nodes in the constructed instance is polynomial in $p$, $q$ for each fixed $m$. 
\end{lemma}
\begin{proof}
The total number of nodes in layer $k$ is 
\begin{equation*}
  3q +  6pq^{k-1}p^{u-k} + \sum_{\bar{k}\in[u];\bar{k}\neq k} (3q + 6pq^{\bar{k}}p^{u-\bar{k}}) < u(3q + 6p^{u+1}),
\end{equation*}
which is polynomial for constant $u=m/3$.
\end{proof}

\begin{lemma}\label{lemmayesno}
Consider only the following construction: The first column contains a $Z$- and $Y$-sub-block. These are connected to a $T$-sub-block in the second column. This $T$-sub-block is connected to an $X$-sub-block in the third column, and also to a dummy sub-block with the same size as the $Y$-sub-block. Then it holds that the corresponding 3DM instance is a YES instance if and only if it is possible to match all nodes of the $Z$-\sbb through the $T$-\sbb with the $X$-sub-block.
\end{lemma}
\begin{proof}
Let 3DM be a YES instance, and let $T'\subseteq T$ be the choice of tupels in an optimal solution. We can build a feasible solution to the corresponding MBA problem by matching the pairs of $X$ and $Z$ elements contained in $T'$ through the respective $T$-sub-block nodes. Remaining nodes in the $Y$-sub-block are matched with the remaining nodes in the $T$-\sbb and the dummy-sub-block.

On the other hand, if the MBA instance allows a feasible matching of all nodes of the $Z$- and $X$-sub-block, then one can create a feasible solution to the 3DM instance by only choosing the corresping nodes of the $T$-\sbb that are traversed.
\end{proof}

\begin{example}[continued]
Figure~\ref{plot7} shows the construction described in Lemma~\ref{lemmayesno} for our example 3DM instance. 

\begin{figure}[htbp]
\begin{center}
\includegraphics[scale=0.5]{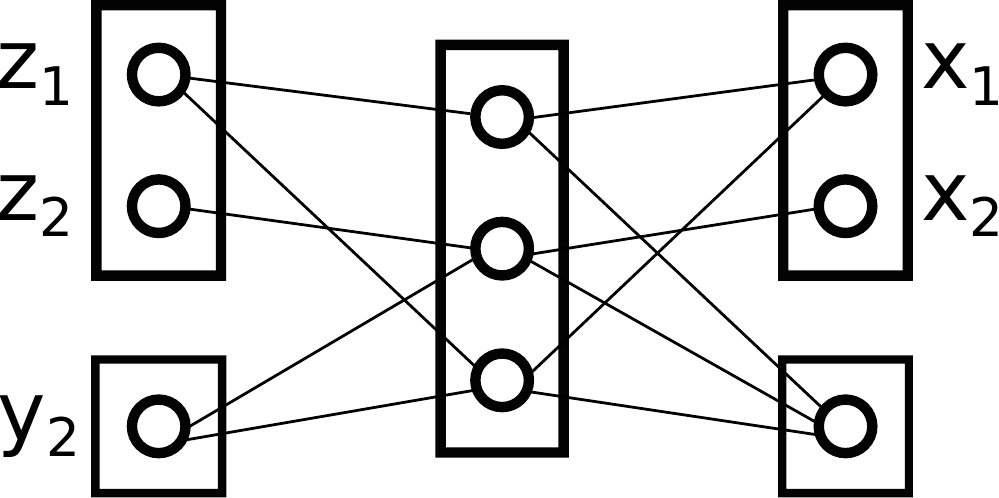}
\end{center}
\caption{Construction in Lemma~\ref{lemmayesno}.}\label{plot7}
\end{figure}

Note that the 3DM instance is indeed a YES instance, and it is possible to match the $Z$- and $X$-\sbb using $t_1$ and $t_2$ in this example.
\end{example}

\begin{lemma} \label{identical_matching_lemma}
In every layer the following is true for any feasible solution of the MBA instance: the matching between head $T$-sub-blocks in $V_2^{\#,j}$ and $V_3^{\#,j}$; and tail $T$-sub-blocks in $V_2^{\#,j}$ and $V_3^{\#,j}$ is exactly identical in every block $j$ of $G_{\#}$, where $\#$ is 0 in the first layer, and 1 in every other layer.
\end{lemma}
\begin{proof}
First we observe that the degree of each node in $V_2^{\#,j}$ and $V_3^{\#,j}$ in the bipartite graph with node-sets $\bigcup_j V_2^{\#,j}$ and $\bigcup_j V_3^{\#,j}$ is equal to 2. Since the $X$-\sbb in $S_{3k}$ is only connected to the head $T$-\sbb of $V_2^{\#,1}$ and two $T$-sub-blocks are connected node-wise, this implies that 
\begin{itemize}
\item the head $T$-\sbb in $V_3^{\#,j}$ is connected to the tail $T$-\sbb of $V_2^{\#,j}$ in the same way in every $j$,
\item the tail $T$-\sbb in $V_3^{\#,j}$ is connected to the head $T$-\sbb of $V_2^{\#,j+1}$ in the same way in every $j< height(G_{\#})$.
\end{itemize} 
\end{proof}

\begin{prop}\label{key_lemma}
If 3DM is a NO instance then in any solution of the resulting MBA instance the following is true: in every layer, there is an $\ell\in [q]$ such that the $\ell^{th}$ node in the $Z$-\sbb in $V_1^{\#,j}$ in  $G_{\#}$ is matched with the head $T$-\sbb in $V_3^{\#,j}$ in  $G_{\#}$, for all $j=1, \ldots, height(G_{\#})$. 
\end{prop}
\begin{proof}
Let $\Gamma$ be the nodes in head $T$-\sbb in $V_3^{\#,j}$ which are matched with the head $T$-\sbb  in $V_2^{\#,j}$. Lemma~\ref{identical_matching_lemma} shows that the sets $\Gamma$ are the same in each $j$. Since it is a NO instance there must exist a $z\in Z$ such that all triple nodes intersecting $z$ are contained in $\Gamma$. If not, this implies that all nodes in the $Z$-\sbb in $V_1^{\#,1}$ can be fully matched with the $X$-sub-block, which is contradiction by Lemma \ref{lemmayesno}. Since the node corresponding to $z$ in the $Z$-\sbb of $V_1^{\#,j}$ is only connected to the head $T$-\sbb in $V_2^{\#,j}$ and only to nodes of $\Gamma$, the statement follows. 
\end{proof}

We denote by $\zeta^{j,\ell}_k$ the set of $z$-nodes in layer $k$ over all groups, which are within block $j\in[q^{k-1}]$ in position $\ell\in[q]$. 
For example in Figure~\ref{zex}, the set $\zeta^{2,1}_2$ denotes those nodes in the middle layer that are in the first position of the second block in each group.

\begin{prop}\label{key_No_prop}
If 3DM is a NO instance then the following is true for the resulting MBA instance: in any partial solution in sets combining layers $1, \ldots, k$, there is an $\ell \in [q]$ and $j\in [q^{k-1}]$, such that the weight of the $3k$-tuples containing nodes of set $\zeta^{j,\ell}_k$ is equal to $k+1$.
\end{prop}

\begin{proof}
Following Proposition \ref{key_lemma}, there is an $\ell \in [q]$ such that in every $j\leq height(G_{\#}$), the $\ell^{th}$ node is connected to head $T$-\sbb of $V_3^{\#,j}$. The proof is by induction. Suppose that the statement is true up to layer $k-1$, that is, there exists a block in the lowest section of layer $k-1$ with at least one element in $Z$-\sbb in a $3(k-1)$-tuple with weight $k$. Since head $T$-\sbb of $V_3^{\#,j}$ is only connected to $\zeta$ sets constructed from $Z$-sub-blocks in layer $k-1$, one of which by construction has all nodes in $3(k-1)$ tuples with weight $k$; and weight of each node in $Z$-sub-blocks of $V_1^{\#,j}$ is 1, it follows that at least one node in some block of lowest section of layer $k$ will be in a $3k$-tuple with weight $k+1$.

\smallskip

It remains to show that statement is true in the base case, that is, for layer 2. Since the all nodes in $Z$- and head $T$-sub-blocks in each block of $G_0$ have weight 1, from Proposition~\ref{key_lemma}, the $\ell^{th}$ node in $Z$-\sbb in each block of $G_0$ is in a triple with weight $2$ and now the statement for base case follows by construction of $\zeta$ sets. 
\end{proof}

\begin{prop}\label{main_results}
If the instance of 3DM is a YES instance, then there exists a solution in the corresponding MBA instance with weight equal to $1$.

If the instance of 3DM is a NO instance, then any solution in the corresponding MBA instance has weight equal to $u+1$.
\end{prop}
\begin{proof}
\textit{YES case:} To prove the statement we need to show that each 1 is not matched with another 1 in the same $3k$-tuple. By Lemma \ref{lemmayesno}, in every layer the $Z$-\sbb in $V_1^{\#,1}$ is fully matched with the $X$-sub-block. Recall that matching within gadgets is completely defined by the matching between the $Z$-\sbb in $V_1^{\#,1}$ and the $X$-\sbb by construction. Using Lemma \ref{lemmayesno}, we match 
\begin{itemize}
    \item the $X$-\sbb in $S_{3k}$ of layer $k$, for all $k>1$, to the dummy block in $S_{3(k-1)-2}$, 
    \item all nodes in the tail $T$-\sbb in $S_{3k}$ of layer $k$, $k>1$, to the dummy block in $S_{3(k-1)-2}$,
    \item all nodes except the $Z$-sub-block in $S_{3k-2}$, $k<u$, to the dummy block in $S_{3(k+1)}$.
\end{itemize}
To complete the proof it is enough observe that each sub-block in $G_{\#}$ is matched with a matching \sbb in the adjacent layer and the $X$-sub-blocks are matched only with dummy blocks in this matching. 

\bigskip

\textit{NO case:} From Proposition \ref{key_No_prop}, it follows that there exists a $3u$-tuple with weight $u+1$. 
\end{proof}

Using the  above (gap) reduction and Proposition~\ref{main_results}, we can now state the following result

\begin{theorem}
MBA with $m=3u$ cannot be approximated to within a ratio $u+1$ unless P=NP.
\end{theorem}


\begin{remark}
Note that, similar to \cite{DKS2012}, the MBA instance constructed here has weights only from $\{0,1\}$. Clearly, this is the simplest possible weight set which implies the hardness of the problem originates mainly from the edge-set of the underlying graph. This observation was already made in \cite{DKS2012} in the case when $m=3$. In fact it is easy to see that when weights are taken from $\{0,1\}$ then the PTAS for the case when edge sets are complete can be used to construct a PTAS for the case when every second edge set is complete. We conjecture that a similar approach can be used with a general weight set. Since our interest in the paper is in the complexity of the general case we leave out the study of these special cases for future study. 
\end{remark}


\section{Solution methods}
\label{sec:methods}

We now consider practical solution methods for the arbitrary MBA. In the following, we refer to the weights in sets as a weight matrix and denote it as $w_{ij}$ for $i\in[n]$, $j \in[m]$. Let $\delta^+(i,j)$ for $(i,j)\in[n]\times[m-1]$ contain those elements $(k,j+1)$ that can be reached from $(i,j)$.

\subsection{Integer programming formulation}
\label{sec:methods1}

We introduce a binary variable $x_{ijk}=1$ for each $i,k\in[n]$ and $j\in[m]$ that is equal to one if and only if tuple $k$ contains element $i$ in level $j$. Let $D$ be a variable modeling the largest weight over all tuples. Then, the MBA can be written as the following integer program:
\begin{align}
\min\ & D \label{ip1}\\
\text{s.t. } & \sum_{i\in[n]} x_{ijk} = 1 & \forall j\in[m], k\in[n] \label{ip2}\\
& \sum_{k\in[n]} x_{ijk} = 1 & \forall i\in[n], j\in[m] \label{ip3}\\
& \sum_{i\in[n]} \sum_{j\in[m]} w_{ij} x_{ijk} \le D & \forall k\in[n] \label{ip4}\\
& \sum_{(i',j+1) \in \delta^+(i,j)} x_{i',j+1,k} \ge x_{ijk} & \forall i,k\in[n],j\in[m-1] \label{ip5}\\
& x_{i1i} = 1 & \forall i\in[n] \label{ip6}\\
& x_{ijk} \in\{0,1\} & \forall i,k\in[n], j\in[m] \label{ip7}
\end{align}
The Objective~\eqref{ip1} is to minimize the largest weight. Constraints~\eqref{ip2} and \eqref{ip3} ensure that each tuple uses exactly one element from each set, and each element in each set is used exactly once, respectively. By Constraint~\eqref{ip4}, we enforce $D$ to become equal to the largest weight over all tuples in an optimal solution. Constraint~\eqref{ip5} model the connectivity of the graph: If $(i,j)$ is used, then it is only possible to use elements $(i',j+1)$ that are in $\delta^+(i,j)$. Finally, the purpose of Constraint~\eqref{ip6} is to break the symmetry in the solution variables.

\subsection{Greedy method}
\label{sec:methods2}

The greedy method has proved highly effective for the complete MBA, and can be directly extended to the arbitrary case. The idea is to construct a solution layer by layer, trying to balance the weights of the tuples in each iteration.

Apart from transferring this idea to the arbitrary case, we also introduce two additional features: Firstly, we do not construct the solution layer by layer, but include a lookahead to further improve the solution quality. Secondly, we present a post-optimization step.

The greedy with lookahead $L\in[m]$ works as follows. We begin with (partial) tuples $P_i = \{ (i,0) \}$ and weights $W_i = w_{i0}$. Then, we solve an IP for $j=1$ up to $j=m-1$, which aims at extending the current set of tuples by one layer, taking the following $L$ layers into account. The IP can be formulated as follows:
\begin{align*}
\min\ & D \\
\text{s.t. } & \sum_{i\in[n]} x_{i,l+1,k} = 1 & \forall l\in[L+1], k\in[n] \\
& \sum_{k\in[n]} x_{i,l+1,k} = 1 & \forall  l\in[L+1],i\in[n] \\
& W_i + w_{i,j+l+1}x_{i,l+1,k} \le D & \forall i,j\in[n], l\in[L+1] \\
& \sum_{(i',j+l+1) \in \delta^+(i,j+l)} x_{i',l+1,k} \ge x_{i,l,k} & \forall i,k\in[n],l\in[L+1] \\
& x_{i1i} = 1 & \forall i\in[n] 
\end{align*}
The constraints and variables are analogous to model (\ref{ip1}--\ref{ip7}), with the difference that layers up to $j$ are already fixed and only the next $L+1$ layers are taken into account. This means that problems are considerably smaller and easier to solve for $L\ll m$. Note that if $j+L+1>m$, the model needs to be adjusted to include less layers.

Once an optimal solution has been determined in iteration $j$, all tuples and their weights are updated and the model for the next layer is solved. Note that while $L$ layers are taken into account during the optimization, we use a rolling horizon where only the result for the first unfixed layer is used. The remaining part of the optimization problem is only used as a lookahead for better balancing the expected weights of the current tuples. 

Furthermore, for $L=0$ this approach is equivalent to the standard greedy method, extended to the arbitrary MBA. Note that in this case, the problem in each iteration corresponds to a bottleneck assignment problem, which can be solved in polynomial time. For $L\ge1$, the subproblem that needs to be solved in the greedy procedure is $NP$-hard, as it is an MBA in itself.

Once a complete set of tuples has been constructed this way, we perform an additional post-optimization step. To this end, we re-optimize the current solution by keeping all layers except the connection between columns $j$ and $j+1$ fixed, and letting $j$ run from $1$ to $m-1$. If it is indeed possible to improve the current solution, the process is repeated until no more improvements are found.

\subsection{Column generation}
\label{sec:methods3}

We now introduce a heuristic solution method of MBA that is based on column generation. To this end, consider the following extended formulation, where $P$ denotes the set of all tuples and a tuple chooses exactly one node $(i,j)$ per layer $j$.
\begin{align*}
\min\ & D \tag{Master} \\
\text{s.t. } & \sum_{k\in[n]} \sum_{p\in P : (i,j)\in p} x_{kp} \ge 1 & \forall i\in[n],j\in[m] \\
& D \ge \sum_{p\in P} w_p x_{kp} & \forall k\in[n]  \\
& x_{kp}\in\{0,1\}
\end{align*}
As before, $D$ is a variable modeling the worst-case weight over all tuples. Variables $x_{kp}$ are used to model whether the $k$th tuple of the solution contains tuple $p\in P$. Note that this model contains a polynomial number of constraints, but a possible exponential number of variables. To avoid this problem, we start with a subset $P'$ of all possible tuples (e.g., representing a starting solution constructed using the greedy heuristic). We iteratively construct additional tuples until the LP-relaxation of problem~(Master) is solved to optimality. To reach an optimal solution to the integer problem (Master), a branching procedure would need to follow. We restrict our column generation approach by only applying it in the root node, i.e., once the LP-relaxation of (Master) has been solved, we solve a restricted integer (Master) problem where all tuples $P'$ generated so far are taken into account.

To find columns, we solve the dual of the restricted master problem to find columns with negative reduced costs. These subproblems are of the form
\begin{align*}
\max\ & \sum_{i\in[n]} \sum_{j\in[m]} u_{ij} \tag{Sub}\\
\text{s.t. } &\sum_{(i,j)\in p} u_{ij} \le w_p r_k & \forall k\in[n],p\in P' \\
&\sum_{k\in[n]} r_k \le 1 \\
&u_{ij} \ge 0 & \forall i\in[n],j\in[m] \\
&r_k \ge 0 & \forall k\in[n]
\end{align*}
Let $u^*_{ij}$ for all $i\in[n]$, $j\in[m]$ and $r^*_k$ for all $k\in[n]$ be an optimal solution to (Sub), and set $k^* = \arg\min_{k\in[n]} r^*_k$. Our aim is to produce a set of tuples $p$ such that their reduced costs $\sum_{(i,j)\in p} (r^*_{k^*}w_{ij}-u^*_{ij})$ is as small as possible. To ensure that these new columns can be combined with other columns, our aim is to produce sets of $n$ tuples in every pricing iteration. This means we can simply use the greedy method from the previous section with modified weights to solve the pricing problem heuristically.

\section{Computational experiments}
\label{sec:experiments}

\subsection{Setup}

The aim of these experiments is to evaluate the quality of the methods presented in Section~\ref{sec:methods} in comparison to the standard greedy method as a representative for the state of the art.

To this end, we generate random instances in the following way. Given $n$ and $m$, a weight matrix is generated by sampling each $w_{ij}$ for $i\in[n]$, $j\in[m]$ uniformly randomly from $\{1,\ldots,100\}$. To create the graph structure, we first generate all horizontal arcs, i.e., arcs connecting $(i,j)$ with $(i,j+1)$ to ensure that a feasible solution exists. We then create tuples from the first to the last layer by choosing random nodes in each layer. Arcs along these tuples are added to the graph, if they don't already exist. For a density parameter $d\ge 0$, we create $dn$ such tuples.

In all applications of MBA, see Section~\ref{lit_review}, typically it is the case that $n>m$. We follow this convention in our instance generation.  We generate 100 random problems for each configuration of $(n,m)\in \{(10,5), (30,8), (35,9), (40,10), (100,15)\}$ using $d=1.8$ and $d=2.2$ (a total of 1000 instances). The choice of parameters is to reflect small-scale problems ($n=10$), medium-scale problems ($n=30,35,40$) and large-scale problems ($n=100$). 

Each instance is solved using the standard greedy method. The resulting solution is used as a baseline comparison. It is included as a starting solution to the IP model from Section~\ref{sec:methods1}. Additionally, we calculate greedy solutions with lookahead for $L=1,2,3$ and post-optimization as described in Section~\ref{sec:methods2} (denoted as GP1, GP2, and GP3). Finally, we also use the column generation approach from Section~\ref{sec:methods3} using the lookahead and post-optimize greedy for the starting solution and the pricing problem (denoted as CG1, CG2, CG3). For large-scale problems, we only use $L=1,2$.

For each method and instance, we allow a time limit of 5 minutes. This can be exceeded for the column generation approach, as we let the final pricing iteration before hitting the time limit complete, and then still solve the restricted master problem. We solve IPs using Cplex version 12.8 and setting the MIPEmphasis parameter such that it focusses on heuristics (so that Cplex spends less time on proving optimality, and the advantage of heuristic methods that do not need to do so is reduced). All experiments are conducted on a virtual Ubuntu server with ten Xeon CPU 
E7-2850 processors at 2.00 GHz speed and 23.5 GB RAM. All processes are restricted to one thread.

\subsection{Results}

For each instance and method, we calculate the ratio of the objective value achieved by greedy and the objective value achieved by the respective method. This means the higher these numbers, the larger is the improvement over the baseline. Table~\ref{tab:1} shows the average of these ratios over the 100 instances in each configuration, while Figure~\ref{fig:boxplots} presents more detailled boxplots for this data.

\begin{table}[htb]
\begin{center}
\begin{tabular}{rrr|rrrrrrr}
$n$ & $m$ & $d$ & IP & GP1 & CG1 & GP2 & CG2 & GP3 & CG3 \\
\hline
10 & 5 & 1.8 & \textbf{4.55} & 3.26 & 3.26 & 3.38 & 3.47 & 3.99 & 3.99 \\
10 & 5 & 2.2 & \textbf{8.28} & 5.44 & 5.58 & 6.70 & 6.85 & 7.36 & 7.36 \\
\hline
30 & 8 & 1.8 & \textbf{9.78} & 6.47 & 6.62 & 7.26 & 8.14 & 7.41 & 8.04 \\
30 & 8 & 2.2 & 10.97 & 8.82 & 9.50 & 9.43 & 10.23 & 10.84 & \textbf{11.32} \\
\hline
35 & 9 & 1.8 & 8.85 & 7.36 & 7.69 & 8.23 & 9.50 & 9.03 & \textbf{9.81} \\
35 & 9 & 2.2 & 1.51 & 8.91 & 9.51 & 9.66 & 10.40 & 11.28 & \textbf{11.33} \\
\hline
40 & 10 & 1.8 & 1.84 & 7.15 & 7.52 & 8.64 & \textbf{9.90} & 8.79 & 9.55 \\
40 & 10 & 2.2 & 0.00 & 10.53 & 10.73 & 11.46 & 11.53 & \textbf{12.27} & \textbf{12.27} \\
\hline
100 & 15 & 1.8 & 0.00 & 10.01 & 10.01 & \textbf{10.03} & \textbf{10.03} & - & - \\
100 & 15 & 2.2 & 0.00 & \textbf{11.31} & \textbf{11.31} & 11.07 & 11.07 & - & -
\end{tabular}
\caption{Average improvement over greedy method in percent. In bold are best values per row.}\label{tab:1}
\end{center}
\end{table}

\begin{figure}[htbp]
\begin{subfigure}[t]{0.45\textwidth}
\includegraphics[width=\textwidth]{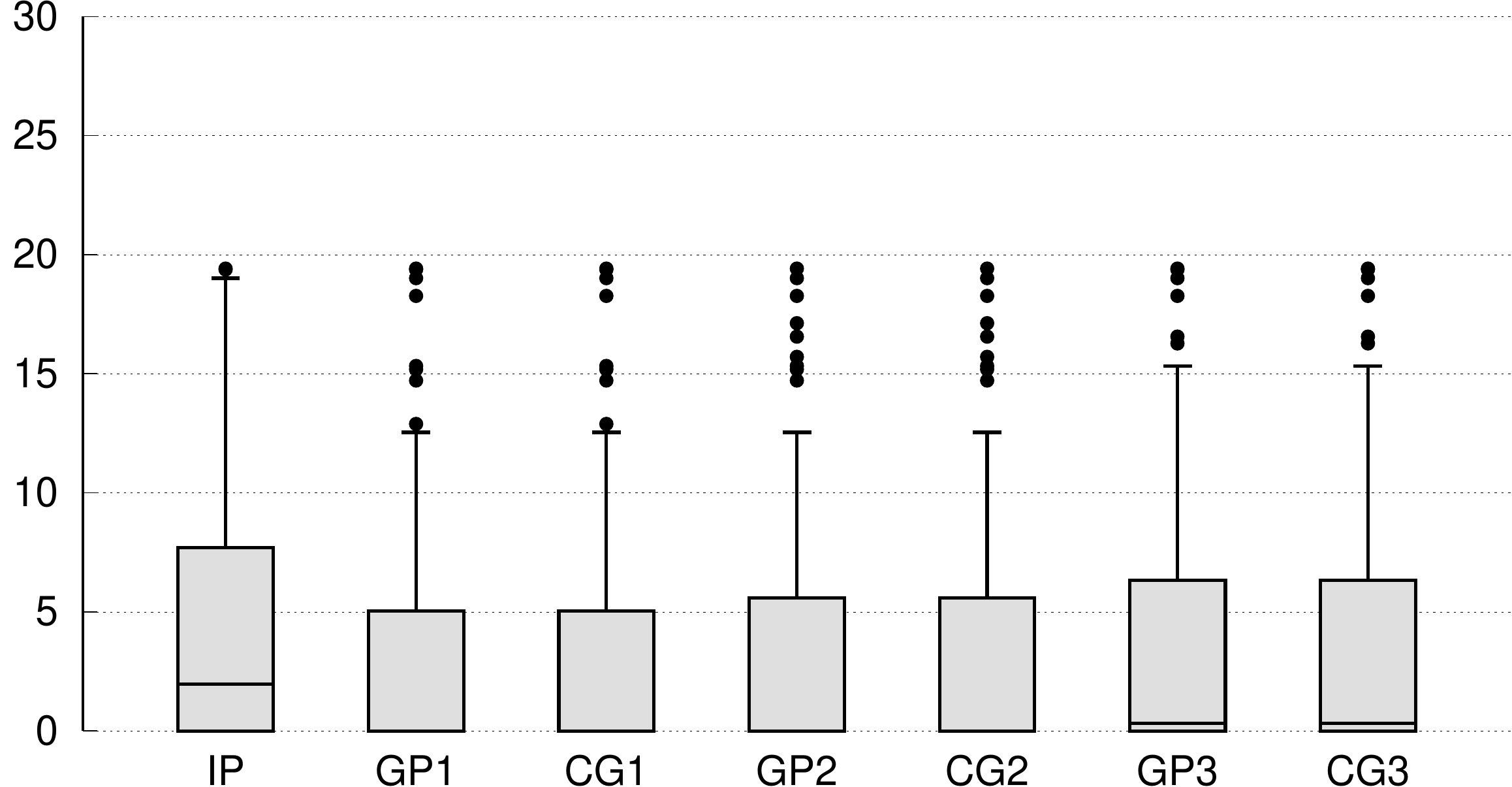}
\caption{$n=10$, $m=5$, $d=1.8$}\label{fig:1}
\end{subfigure}
\hfill
\begin{subfigure}[t]{0.45\textwidth}
\includegraphics[width=\textwidth]{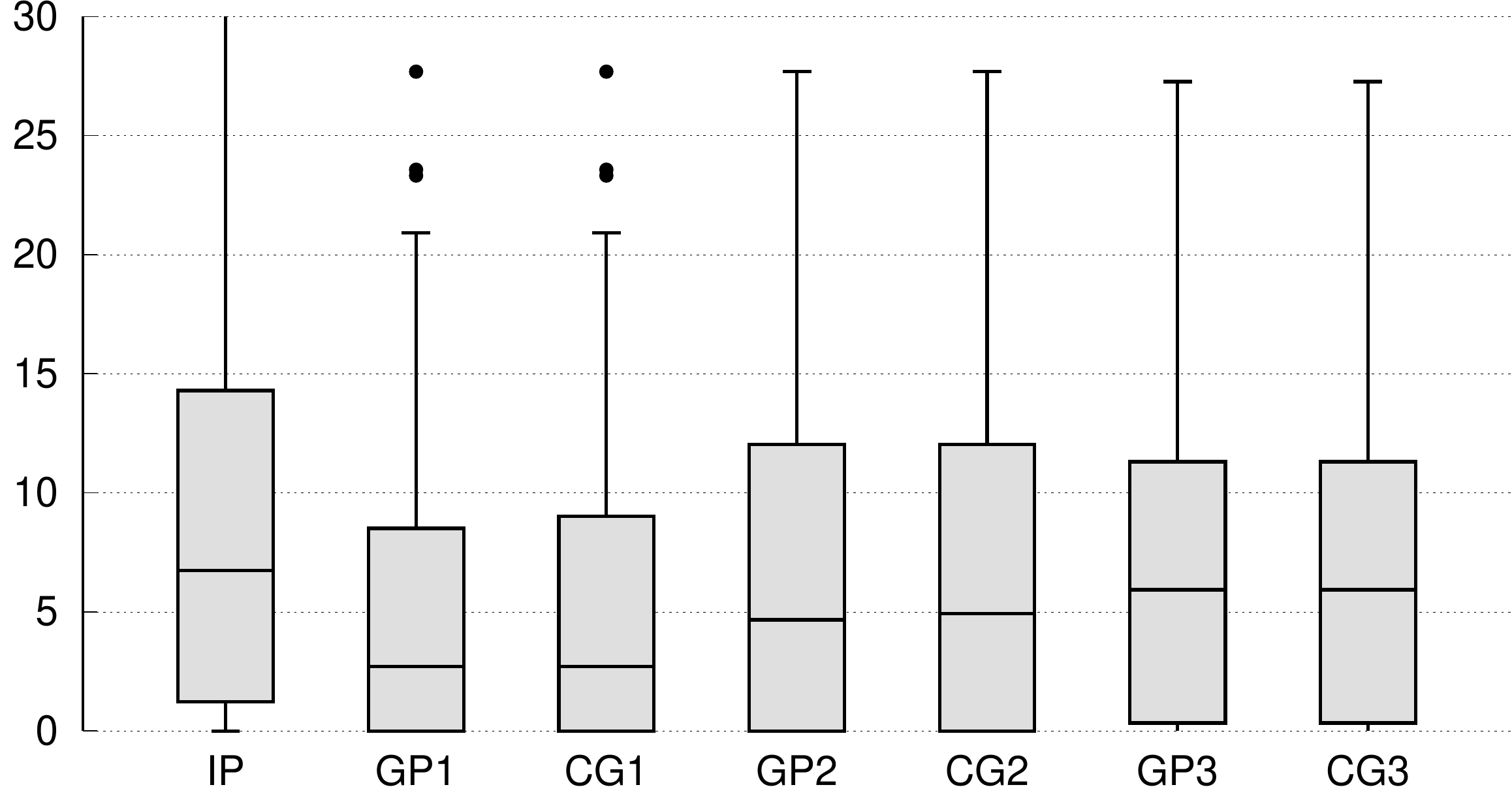}
\caption{$n=10$, $m=5$, $d=2.2$}\label{fig:2}
\end{subfigure}
\begin{subfigure}[t]{0.45\textwidth}
\includegraphics[width=\textwidth]{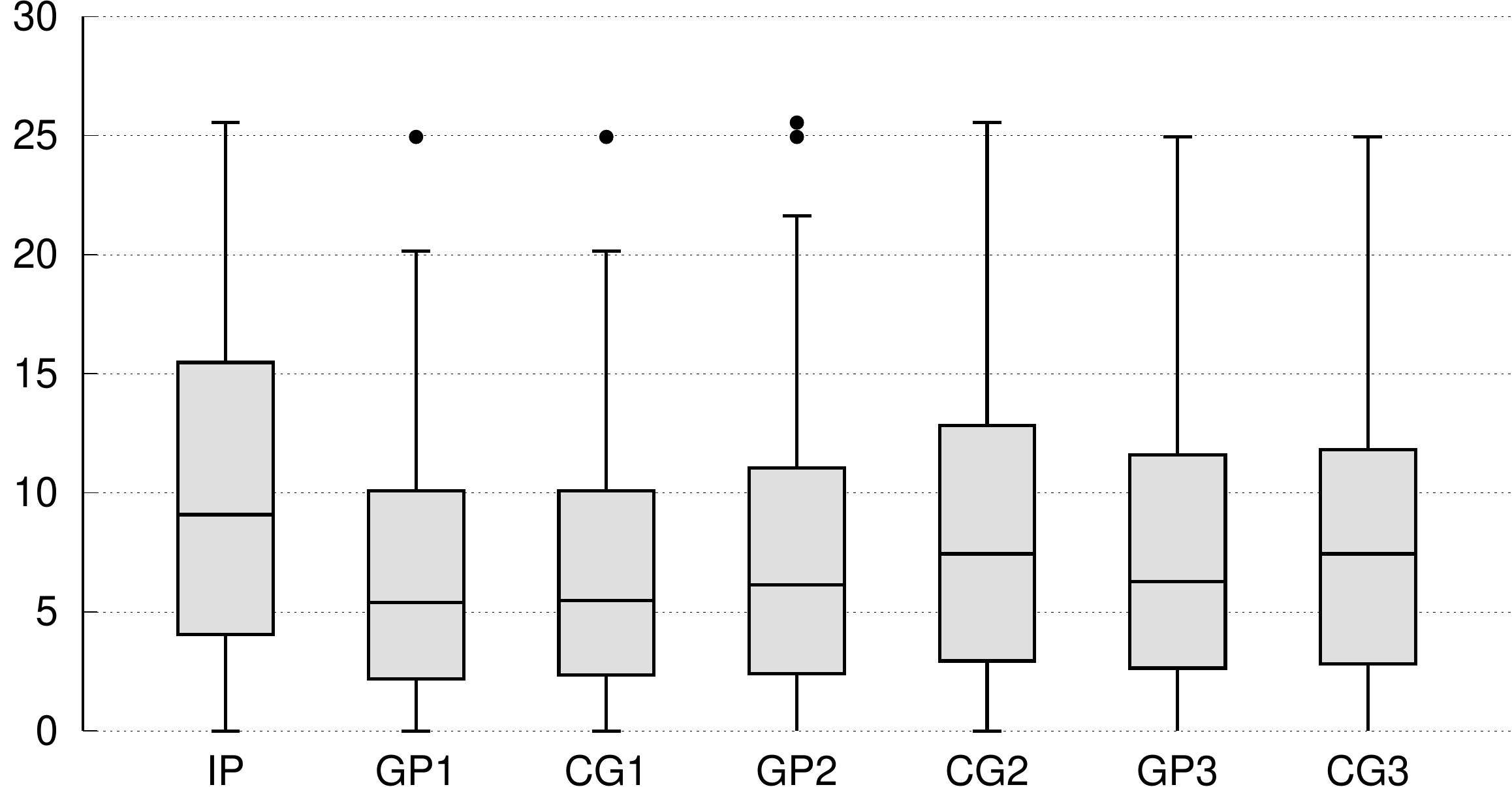}
\caption{$n=30$, $m=8$, $d=1.8$}\label{fig:3}
\end{subfigure}
\hfill
\begin{subfigure}[t]{0.45\textwidth}
\includegraphics[width=\textwidth]{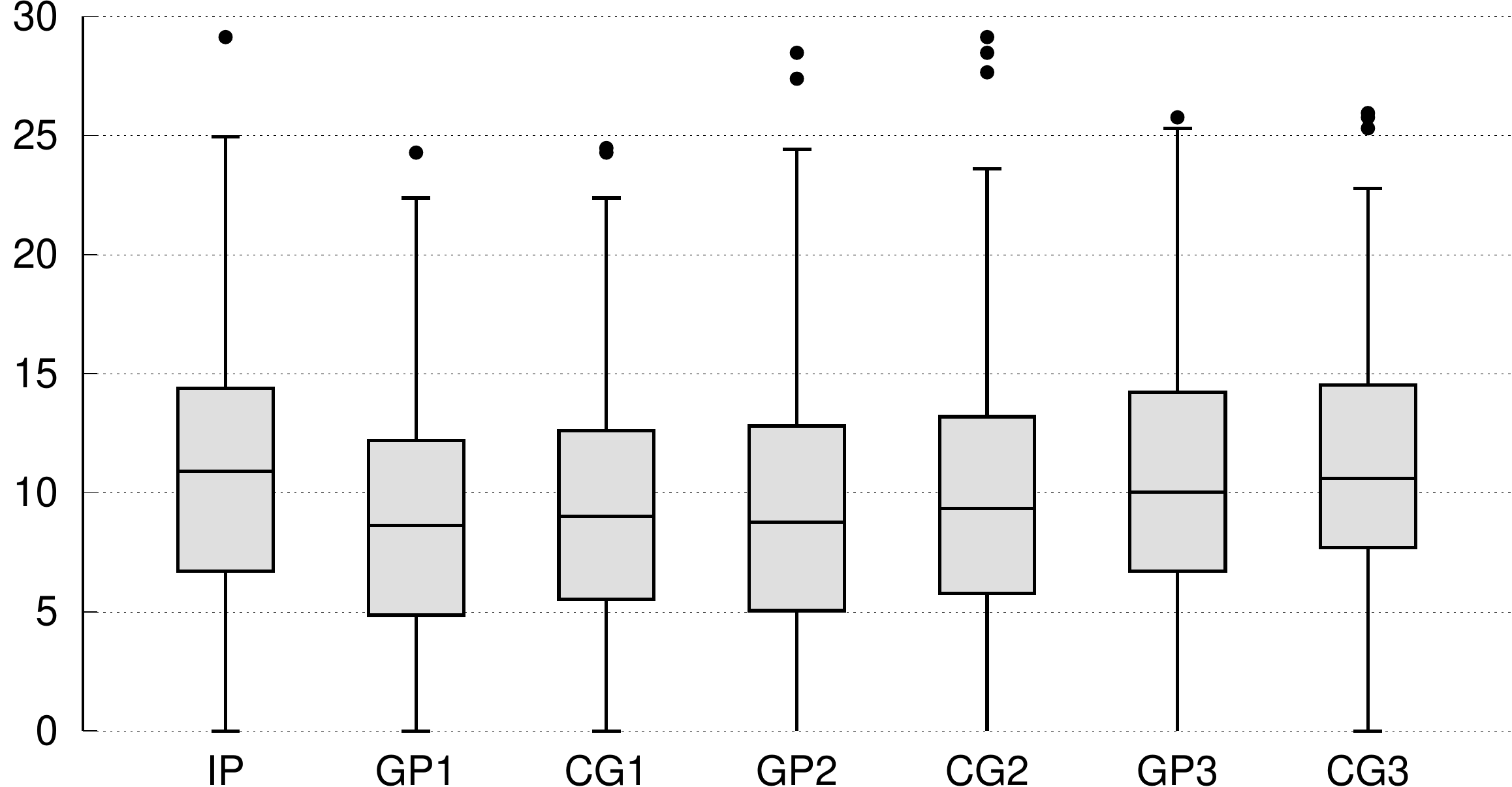}
\caption{$n=30$, $m=8$, $d=2.2$}\label{fig:4}
\end{subfigure}
\begin{subfigure}[t]{0.45\textwidth}
\includegraphics[width=\textwidth]{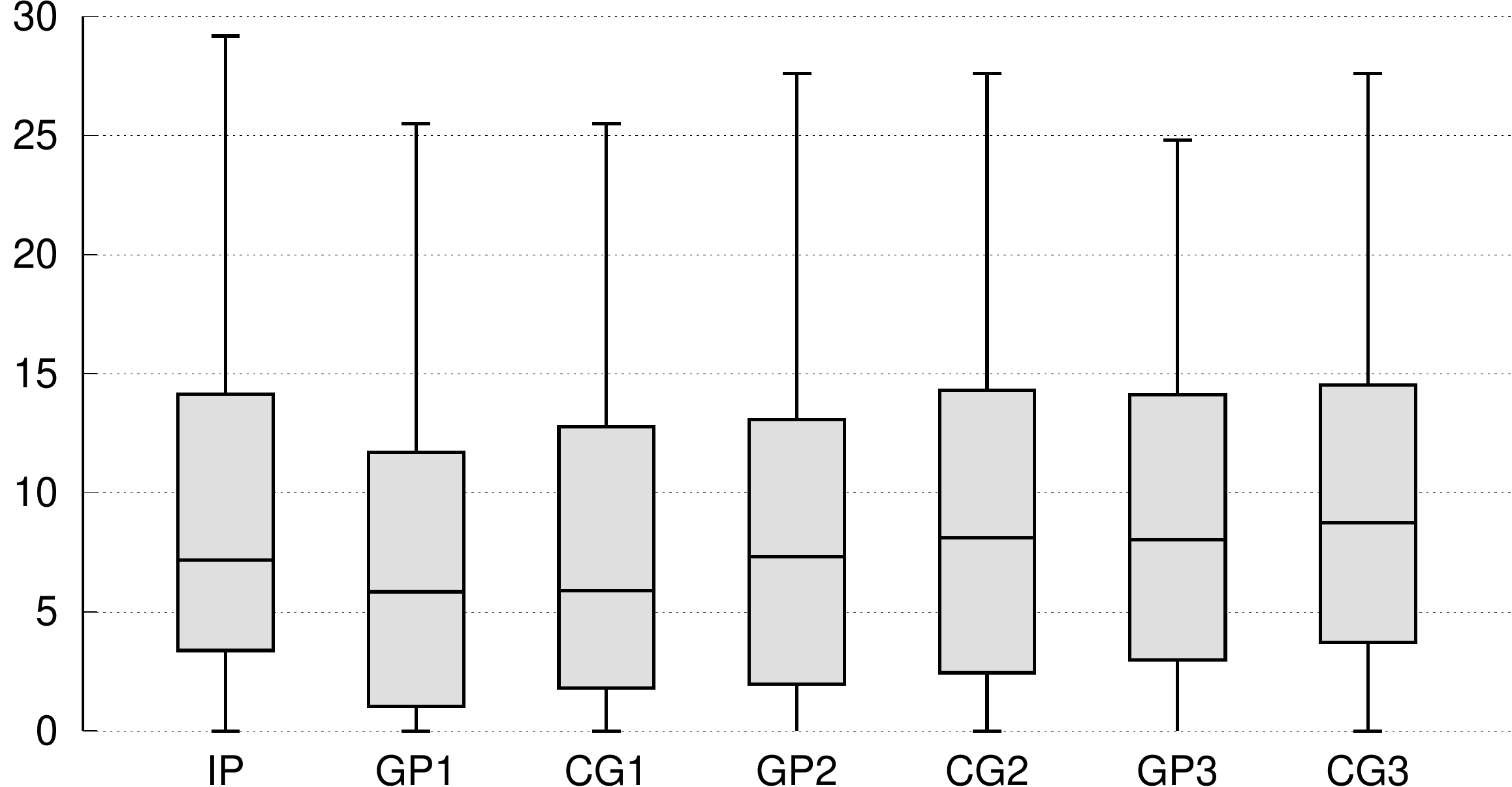}
\caption{$n=35$, $m=9$, $d=1.8$}\label{fig:5}
\end{subfigure}
\hfill
\begin{subfigure}[t]{0.45\textwidth}
\includegraphics[width=\textwidth]{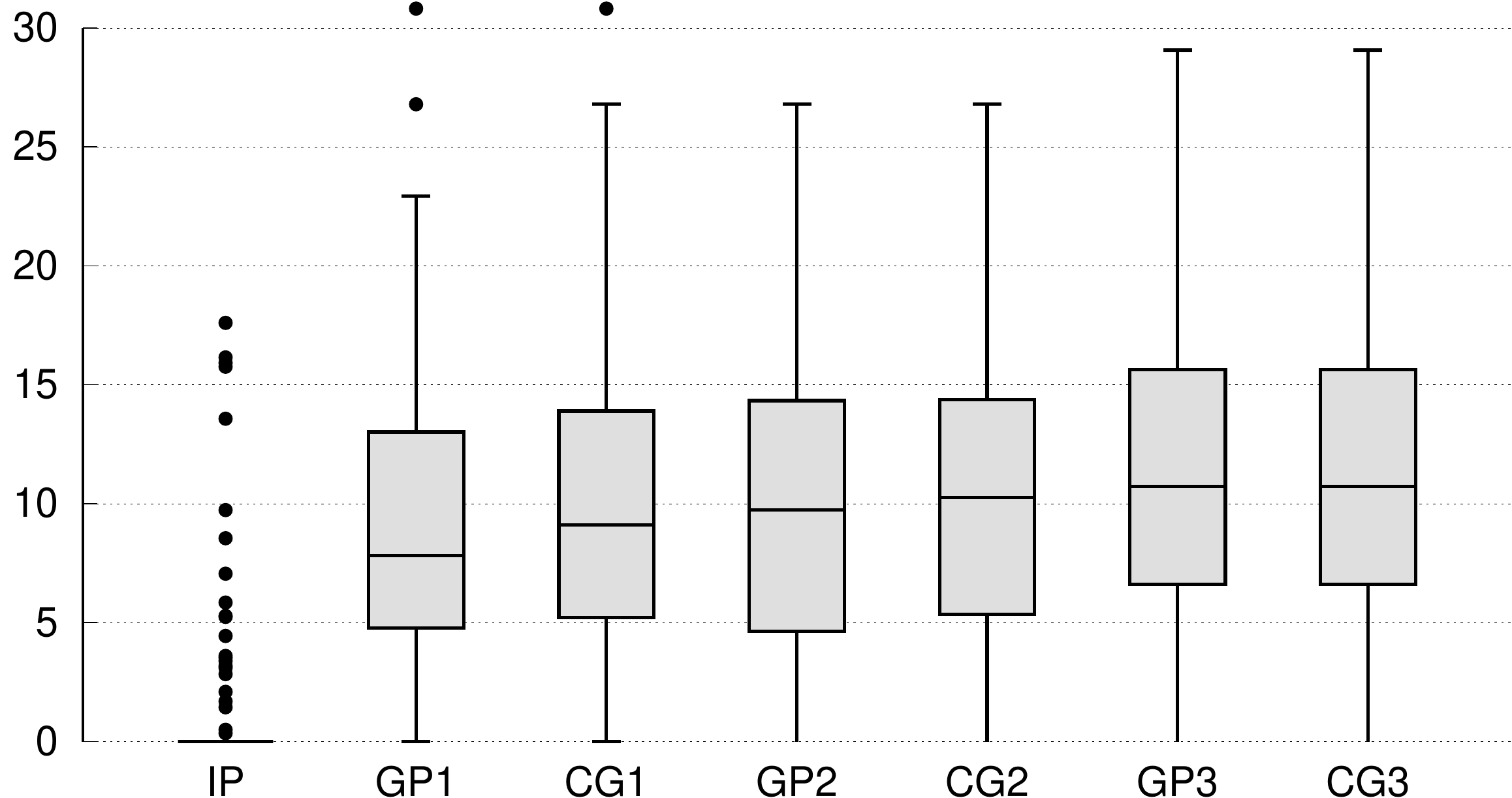}
\caption{$n=35$, $m=9$, $d=2.2$}\label{fig:6}
\end{subfigure}
\begin{subfigure}[t]{0.45\textwidth}
\includegraphics[width=\textwidth]{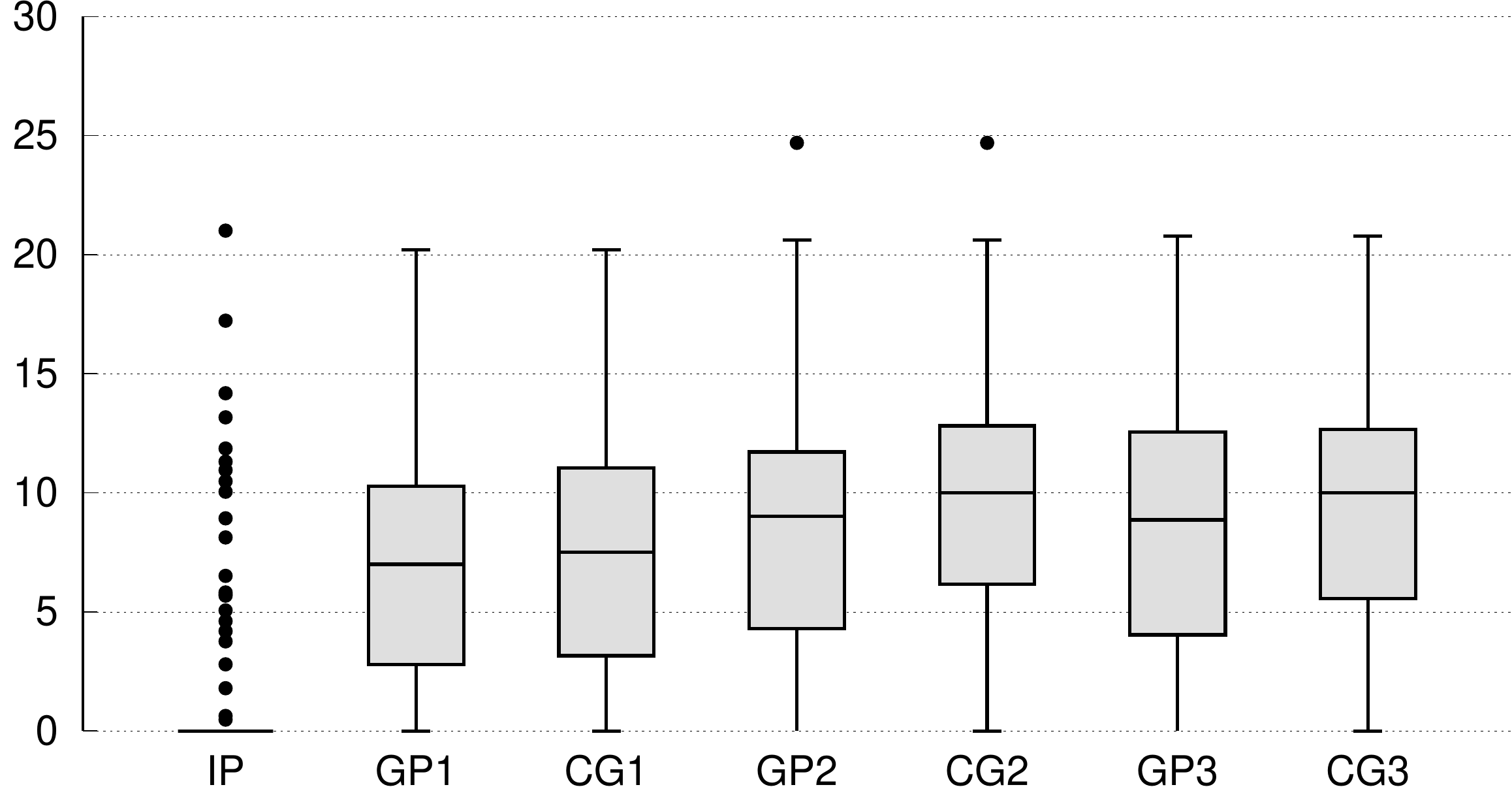}
\caption{$n=40$, $m=10$, $d=1.8$}\label{fig:7}
\end{subfigure}
\hfill
\begin{subfigure}[t]{0.45\textwidth}
\includegraphics[width=\textwidth]{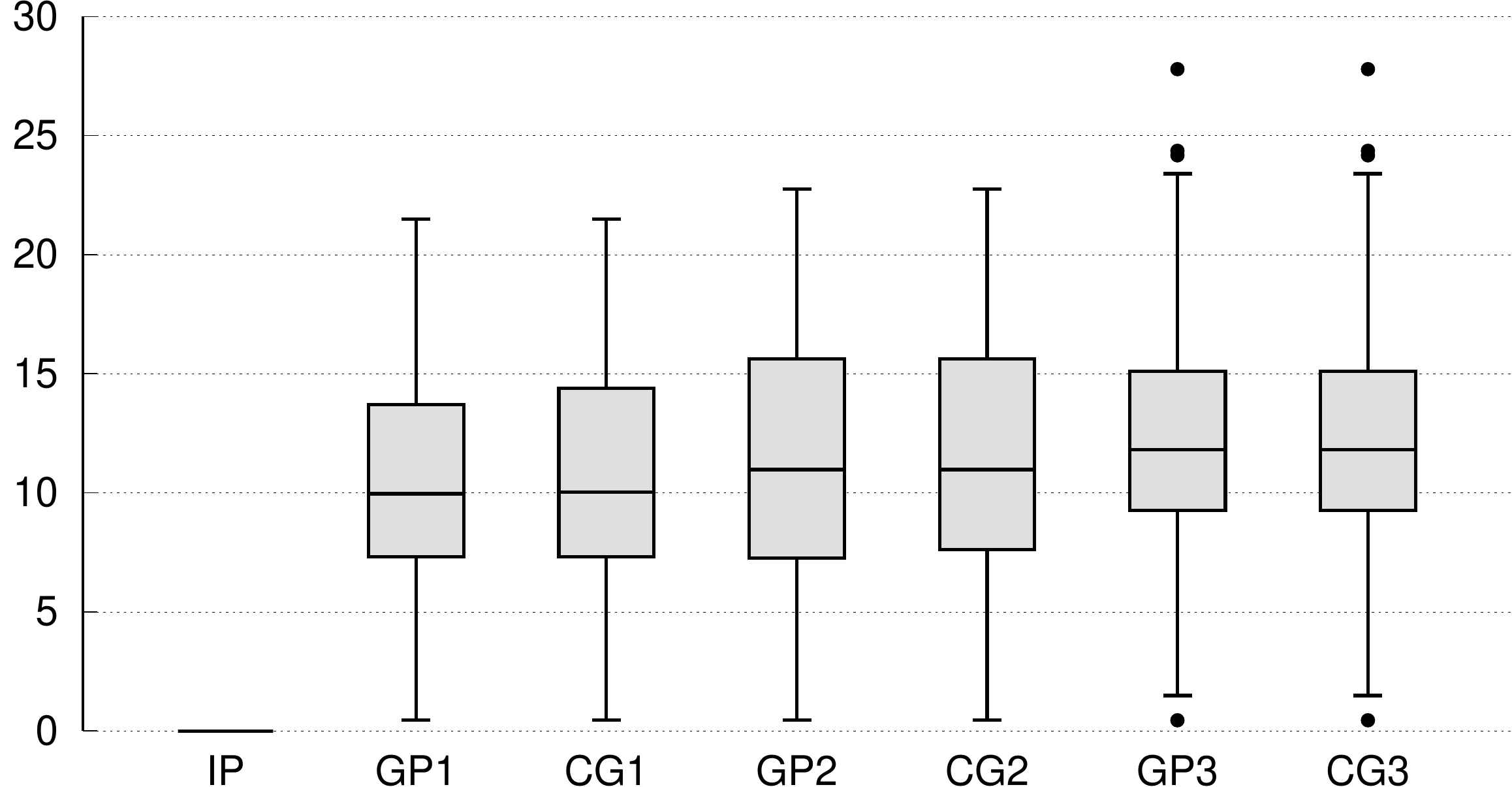}
\caption{$n=40$, $m=10$, $d=2.2$}\label{fig:8}
\end{subfigure}
\begin{subfigure}[t]{0.45\textwidth}
\includegraphics[width=\textwidth]{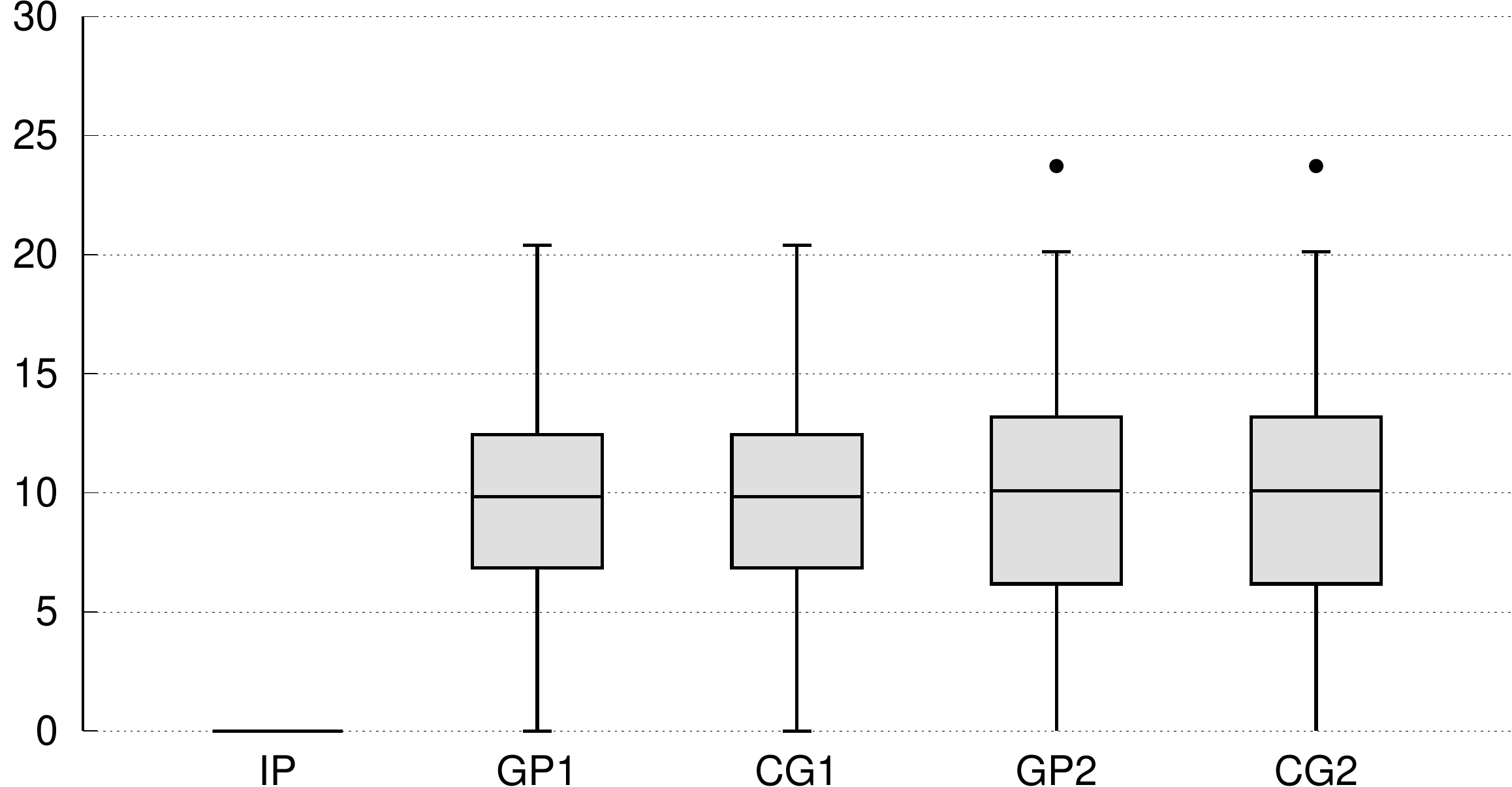}
\caption{$n=100$, $m=15$, $d=1.8$}\label{fig:9}
\end{subfigure}
\hfill
\begin{subfigure}[t]{0.45\textwidth}
\includegraphics[width=\textwidth]{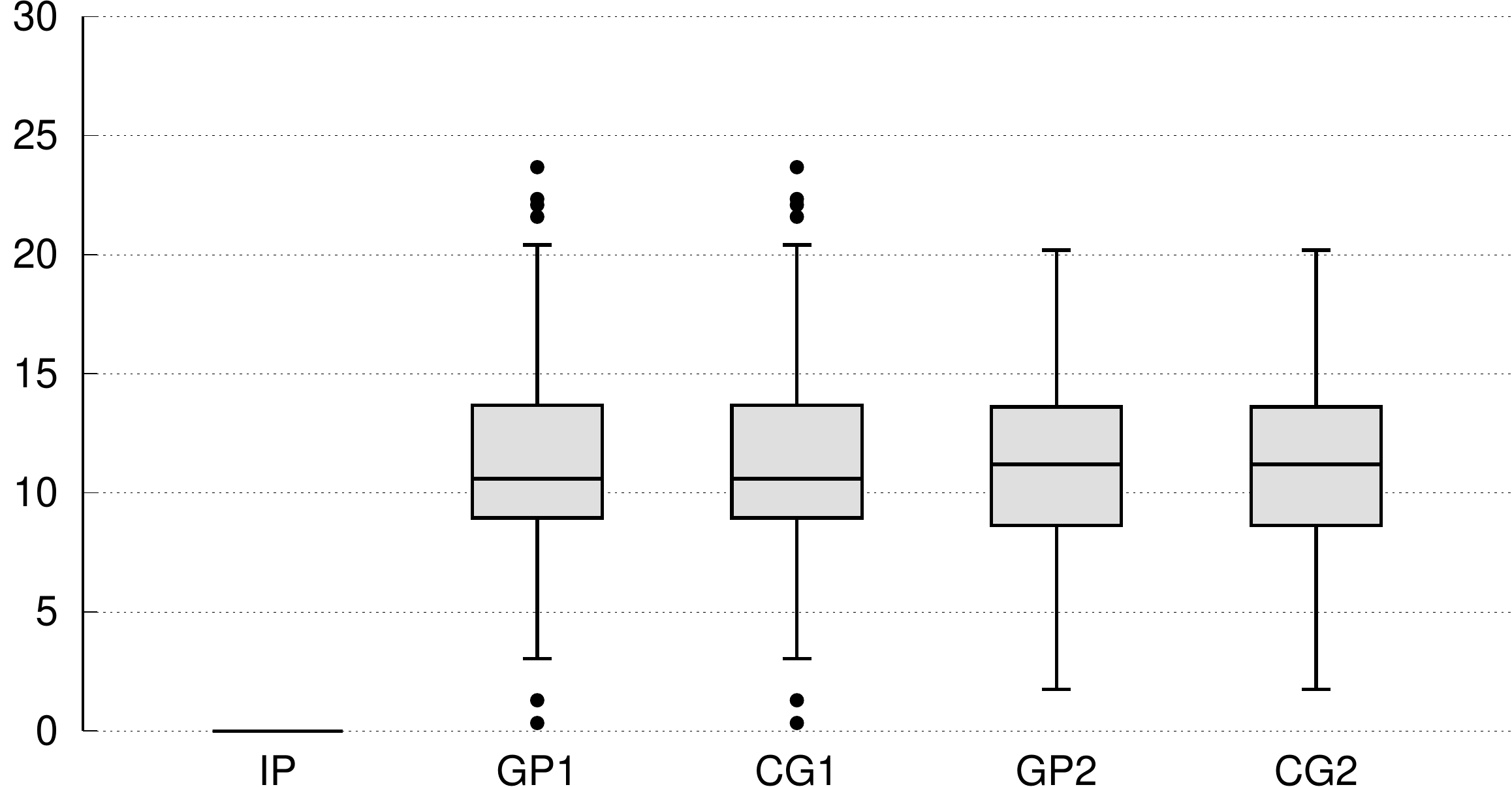}
\caption{$n=100$, $m=15$, $d=2.2$}\label{fig:10}
\end{subfigure}
\caption{Objective value reduction compared to greedy.}\label{fig:boxplots}
\end{figure}

We note that for small-scale problems ($n=10$), the IP approach performs best, i.e., it gives the largest improvement over the baseline. With increasing value of $L$, both GP and CG achieve better results. The advantage of using CG instead of GP is relatively small; only in some cases can the column generation improve the starting solution produced by greedy with lookahead and post-optimization. The improvement over the baseline that can be achieved with higher density ($d=2.2$) in comparison with lower density ($d=1.8$) is larger, i.e., with a larger set of possible solutions, there is also more potential to improve the starting solution.

The performance of methods is very different for medium-sized problems ($n=30,35,40$). Firstly we find that the performance of IP strongly degrades, so far that it is not able to improve the greedy starting solution for $n=40$ and $d=2.2$ at all. While higher density still means that larger improvements can be achieved, it also means that IP models are more complex, and the IP performance is worse for $d=2.2$ than for $1.8$. This is not the case for the other methods. 
Secondly, it is now possible for CG to clearly outperform its GP starting solution. In this region, the column generation approach is able to produce the best solutions found.

For large-scale problems ($n=100$), we find again that IP cannot improve its starting solution, as expected. But in this case, this also applies to CG. Still, the post-optimization makes a significant difference (less so the lookahead), with improvements over 10\% on average in comparison to the baseline.

This comparison of algorithm performance is complemented by Tables~\ref{tab:2} and \ref{tab:3}, which show the average computation times. Table~\ref{tab:2} also indicates the number of instances out of 100 that were solved to proven optimality by the IP approach. Even for large instances, GP1 shows a fast performance with 4.1 seconds on average. The lookahead parameter has a more significant impact on computation times than the post-optimization, with computation times considerably increasing with growing $L$.

\begin{table}[htb]
\begin{center}
\begin{tabular}{rrr|rrrrr}
$n$ & $m$ & $d$ & IP & Greedy & GP1 & GP2 & GP3 \\
\hline
10 & 5 & 1.8 & 0.0 (100) & 0.0 & 0.0 & 0.0 & 0.1\\
10 & 5 & 2.2 & 0.2 (100) & 0.0 & 0.0 & 0.1 & 0.2\\
\hline
30 & 8 & 1.8 & 145.6 (67) & 0.1 & 0.2 & 0.5 & 1.4\\
30 & 8 & 2.2 & 271.6 (15) & 0.1 & 0.3 & 0.8 & 6.2\\
\hline
35 & 9 & 1.8 & 246.9 (25) & 0.2 & 0.3 & 0.8 & 2.4\\
35 & 9 & 2.2 & 296.0 (4) & 0.2 & 0.4 & 1.3 & 12.8\\
\hline
40 & 10 & 1.8 & 289.7 (7) & 0.2 & 0.4 & 1.2 & 4.1\\
40 & 10 & 2.2 & 299.9 (0) & 0.3 & 0.6 & 2.1 & 31.6\\
\hline
100 & 15 & 1.8 & 300.3 (0) & 1.6 & 3.4 & 16.6 & -\\
100 & 15 & 2.2 & 300.3 (0) & 1.8 & 4.1 & 21.8 & -
\end{tabular}
\caption{Average computation time in seconds. For IP, the value in brackets indicates the number of instances solved to provel optimality.}\label{tab:2}
\end{center}
\end{table}

The computation times for CG are separated into the time required for solving the priving problem iteratively (pre), and solving the resulting master problem after columns have been generated (master). Table~\ref{tab:4} indicates the average number of tuples that are produced in the process. This value is connected with the computation time of the resulting master problem. We see that the computational burden of CG is higher than for GP, as can be expected.

\begin{table}[htb]
\begin{center}
\begin{tabular}{rrr|rrr|rrr|rrr}
 &  &  & \multicolumn{3}{c}{CG1} & \multicolumn{3}{c}{CG2}  & \multicolumn{3}{c}{CG3} \\
$n$ & $m$ & $d$ & pre & master & total & pre & master & total & pre & master & total \\
\hline
10 & 5 & 1.8 & 0.1 & 0.0 & 0.1 & 0.1 & 0.0 & 0.1 & 0.2 & 0.0 & 0.2 \\
10 & 5 & 2.2 & 0.2 & 0.0 & 0.2 & 0.4 & 0.0 & 0.4 & 0.6 & 0.0 & 0.6 \\
\hline
30 & 8 & 1.8 & 6.8 & 0.1 & 6.9 & 51.1 & 0.4 & 51.5 & 73.7 & 0.2 & 74.0 \\
30 & 8 & 2.2 & 60.0 & 2.7 & 62.7 & 125.0 & 3.8 & 128.8 & 255.3 & 5.9 & 261.2 \\
\hline
35 & 9 & 1.8 & 16.2 & 0.2 & 16.4 & 134.1 & 1.0 & 135.1 & 184.2 & 0.8 & 185.0 \\
35 & 9 & 2.2 & 149.5 & 10.5 & 160.0 & 253.6 & 44.4 & 298.0 & 341.8 & 12.3 & 354.2 \\
\hline
40 & 10 & 1.8 & 38.1 & 0.4 & 38.4 & 268.3 & 4.4 & 272.7 & 283.9 & 3.5 & 287.4 \\
40 & 10 & 2.2 & 232.5 & 84.0 & 316.5 & 301.4 & 99.6 & 401.0 & 475.8 & 13.9 & 489.7 \\
\hline
100 & 15 & 1.8 & 359.2 & 0.5 & 359.7 & 366.0 & 0.3 & 366.3 & - & - & - \\
100 & 15 & 2.2 & 381.6 & 0.4 & 382.0 & 389.5 & 0.3 & 389.8 & - & - & -
\end{tabular}
\caption{Average computation time in seconds.}\label{tab:3}
\end{center}
\end{table}

\begin{table}[htb]
\begin{center}
\begin{tabular}{rrr|rrr}
$n$ & $m$ & $d$ & CG1 & CG2 & CG3 \\
\hline
10 & 5 & 1.8 & 12.8 & 13.3 & 11.5 \\
10 & 5 & 2.2 & 21.3 & 23.5 & 17.6 \\
\hline
30 & 8 & 1.8 & 180.0 & 411.7 & 345.4 \\
30 & 8 & 2.2 & 512.7 & 711.2 & 718.3 \\
\hline
35 & 9 & 1.8 & 256.4 & 637.3 & 573.0 \\
35 & 9 & 2.2 & 751.9 & 989.7 & 661.9 \\
\hline
40 & 10 & 1.8 & 351.5 & 889.0 & 733.8 \\
40 & 10 & 2.2 & 868.4 & 1060.4 & 526.4 \\
\hline
100 & 15 & 1.8 & 760.3 & 693.0 & 693.0 \\
100 & 15 & 2.2 & 697.7 & 644.0 & -
\end{tabular}
\caption{Average numbers of tuples produced in column generation.}\label{tab:4}
\end{center}
\end{table}

To summarize the experimental results, we see that small-scale problems can be solved using the IP formulation in little time, and heuristic methods are not required in these cases. Starting with $n=30$, heuristic methods become important. By using lookahead and post-optimization, the standard greedy method can be outperformed by around $10\%$ on average, while still retaining small computation times. These improvements can be further strengthened by using the column generation approach, at the cost of higher computation times. Finally, for large-scale problems, the post-optimization step becomes the most important tool to improve greedy while still keeping computation times small.

\section{Conclusions}
\label{sec:conclusions}

We considered the multi-level bottleneck assignment problem, which has attracted considerable attraction in application areas such as finance and scheduling. While previous models allowed that any element in one set can be paired with any other element in the next set, we generalized this setting such that any bipartite graph can represent the feasible pairings between two consecutive columns. We analyzed the complexity of this problem, showing that it is not approximable better than $(\lfloor\frac{m}{3}\rfloor+1)$ in polynomial time. 

A prominent solution method for the complete MBA is a greedy approach, where solutions are built by iterative solving a bottleneck assignment problem. We transferred this method to the arbitrary MBA and further improved it using lookahead and post-optimization. We also introduced alternative methods based on solving an IP formulation, and using a column generation approach. In computational experiments we found that the standard greedy approach can be outperformed considerably.

\end{document}